\documentclass{conm-p-l} %Document class changes to AMS Standard 6/29/12

\usepackage{graphicx}
\usepackage{mathrsfs, amsthm, amsmath, amssymb}
\usepackage{appendix}
\usepackage{bm, hyperref}
\usepackage{cite}
\usepackage{url}
\usepackage{epsfig} %NEW 8/7/11
\usepackage{euscript} %NEW 8/7/11
\usepackage{amsfonts} %NEW 8/7/11

\newtheorem{theorem}{Theorem}[section] %Modified 3/12/12 to fit MASSM paper
\newtheorem{lemma}[theorem]{Lemma}
\newtheorem{corollary}[theorem]{Corollary}

\newtheorem{proposition}[theorem]{Proposition}

\newtheorem{openproblem}[theorem]{Open Problem}

\theoremstyle{definition}
\newtheorem{definition}[theorem]{Definition}

\newtheorem{example}[theorem]{Example}

\newtheorem{remark}[theorem]{Remark}

\newtheorem{notation}[theorem]{Notation}

\newcommand{\C}{\mathbb{C}}
\newcommand{\N}{\mathbb{N}}

\newcommand{\R}{\mathbb{R}}
\newcommand{\Z}{\mathbb{Z}}

\newcommand{\bfPhi}{\mathbf{\Phi}}
\newcommand{\calD}{\mathcal{D}} %NEW 8/7/11
\newcommand{\calL}{\mathcal{L}}
\newcommand{\calM}{\mathcal{M}} %NEW 7/27/11
\newcommand{\diam}{\textnormal{diam}} %NEW 7/31/11, modified 8/3/11
\newcommand{\mink}{\mathscr{M}}

\newcommand{\vol}{\textnormal{vol}}

\begin{document}

\title[Box-Counting Fractal Strings and Zeta Functions]{Box-Counting Fractal Strings, Zeta Functions, and Equivalent Forms of Minkowski Dimension}

%    Information for first author
\author[M.~L.~Lapidus]{Michel~L.~Lapidus}
%    Address of record for the research reported here
\address{Department of Mathematics\\ 
University of California\\ Riverside, California 92521-0135 USA}
%    Current address
%\curraddr{Department of Mathematics and Statistics,
%Case Western Reserve University, Cleveland, Ohio 43403}
\email{lapidus@math.ucr.edu}
%    \thanks will become a 1st page footnote.
\thanks{The work of the first author (M.~L.~Lapidus) was partially supported by the US National Science Foundation (NSF) under the research grant DMS--1107750, as well as by the Institut des Hautes Etudes Scientifiques (IHES)}

%    Information for second author
\author[J.~A.~Rock]{John~A.~Rock}
\address{Department of Mathematics and Statistics\\ California State Polytechnic University\\ Pomona, California 91768 USA}
\email{jarock@csupomona.edu}
%\thanks{Support information for the second author.}

%    Information for third author
\author[D.~\v Zubrini\'c]{Darko \v Zubrini\'c}
\address{Department of Applied Mathematics\\ Faculty of Electrical Engineering and Computing\\ University of Zagreb\\ Unska 3, 10000 Zagreb, Croatia}
\email{darko.zubrinic@fer.hr}
\thanks{The work of the third author (D.~\v Zubrini\'c) was supported by the Ministry of Science of the Republic of Croatia under grant no.\ 036-0361621-1291}

%    General info
\subjclass[2010]{Primary: 11M41, 28A12, 28A75, 28A80. Secondary: 28B15, 37F35, 40A05, 40A10}
\date{July 18, 2012}

%\dedicatory{This paper is dedicated to our advisors.}

\keywords{Fractal string, geometric zeta function, box-counting fractal string, box-counting zeta function, distance zeta function, tube zeta function, similarity dimension, box-counting dimension, Minkowski dimension, Minkowski content, complex dimensions, Cantor set, Cantor string, counting function, self-similar set.}

\begin{abstract}
We discuss a number of techniques for determining the Minkowski dimension of bounded subsets of some Euclidean space of any dimension, including: the box-counting dimension and equivalent definitions based on various box-counting functions; the similarity dimension via the Moran equation (at least in the case of self-similar sets); the order of the (box-)counting function; the classic result on compact subsets of the real line due to Besicovitch and Taylor, as adapted to the theory of fractal strings; and the abscissae of convergence of new classes of zeta functions. Specifically, we define box-counting zeta functions of infinite bounded subsets of Euclidean space and discuss results from \cite{LapRaZu} pertaining to distance and tube zeta functions. Appealing to an analysis of these zeta functions allows for the development of theories of complex dimensions for bounded sets in Euclidean space, extending techniques and results regarding (ordinary) fractal strings obtained by the first author and van Frankenhuijsen.
\end{abstract}

\maketitle

%%%%%%%%%%%%%%%%%%%%%%%%%%%%%%%%%%

\section{Introduction}
\label{sec:Introduction}

Motivated by the theory of complex dimensions of fractals strings (the main theme of \cite{LapvF6}), we introduce box-counting fractal strings and box-counting zeta functions which, along with the distance and tube zeta functions of \cite{LapRaZu}, provide possible foundations for the pursuit of theories of complex dimensions for arbitrary bounded sets in Euclidean space of any dimension. We also summarize a variety of well-known techniques for determining the box-counting dimension, or equivalently the Minkowski dimension, of such sets. Thus, while new results are presented in this paper, it is partially expository and also partially tutorial.

Our main result establishes line (iv) of the following theorem. (See also Theorem \ref{thm:MainResult} below, along with the relevant definitions provided in this paper.) The other lines have been established elsewhere in the literature, as cited accordingly throughout the paper.

\begin{theorem}
\label{thm:Summary}
Let $A$ be a bounded infinite subset of $\R^m$ \emph{(}equipped with the usual metric\emph{)}. Then the following quantities are equal\emph{:}\footnote{The fact that $A$ is infinite is only required for part (iv) of the theorem; see Remark \ref{rmk:FiniteBoxCountingFractalString}.}
\begin{enumerate}
\setlength{\itemsep}{0in}
\item the \emph{upper box-counting dimension} of $A$\emph{;}
\item the \emph{upper Minkowski dimension} of $A$\emph{;}
\item the \emph{asymptotic order of growth of the counting function} of the box-counting fractal string $\calL_B$\emph{;}
\item the \emph{abscissa of convergence} of the box-counting zeta function $\zeta_B$\emph{;}
\item the \emph{abscissa of convergence} of the distance zeta function $\zeta_d$.
\end{enumerate}
\end{theorem}

A summary of the remaining sections of this paper is as follows:
 
In Section \ref{sec:ClassicNotionsOfDimension}, we discuss classical notions of dimension such as \emph{similarity dimension} (or \emph{exponent}), \emph{box-counting dimension}, and \emph{Minkowski dimension} as well as their properties. (See \cite{BesTa,Falc,HaPo,Hut,LapLuvF1,LapLuvF2,LapPe3,LapPe2,LapPeWin,LapPo1,LapRaZu,
LapvF6,MeSa,Mor,Tr2,Tri,Zu,ZuZup}.)

In Section \ref{sec:FractalStringsAndZetaFunctions}, we summarize but a few of the interesting results on fractal strings and counting functions regarding, among other things, geometric zeta functions, complex dimensions, the order of a counting function, and connections with Minkowski measurability. (See \cite{BesTa,LapPo1,LapvF6,Lev,Tri}.) The material in Sections \ref{sec:ClassicNotionsOfDimension} and \ref{sec:FractalStringsAndZetaFunctions} motivates the results presented in Sections \ref{sec:BoxCountingFractalStringsAndZetaFunctions} and \ref{sec:DistanceAndTubeZetaFunctions}.

In Section \ref{sec:BoxCountingFractalStringsAndZetaFunctions}, we introduce \emph{box-counting fractal strings} and \emph{box-counting zeta functions} and, in particular, we show that the abscissa of convergence of the box-counting zeta function of a bounded infinite set is the upper box-counting dimension of the set. These topics are the focus of \cite{LapRoZu}.

In Section \ref{sec:DistanceAndTubeZetaFunctions}, we share recent results from \cite{LapRaZu} on \emph{distance, tube} and \emph{relative zeta functions}, including connections between the corresponding complex dimensions and Minkowski content and measurability. %(See \cite{LapRaZu,Res}.)

In Section \ref{sec:SummaryOfResultsAndOpenProblems}, Theorem \ref{thm:Summary} is restated in Theorem \ref{thm:MainResult} using notation and terminology discussed throughout the paper. We also propose several open problems for future work in this area.

\section{Classic notions of dimension}
\label{sec:ClassicNotionsOfDimension}

We begin with a brief discussion of a classic method for constructing self-similar fractals and a famous fractal, the Cantor set $C$. (See \cite{Falc,Hut}.)

\begin{definition}
\label{def:IFSAndAttractor}
Let $N$ be an integer such that $N \geq 2$. An \emph{iterated function system} (IFS) $\bfPhi=\{\Phi_j\}_{j=1}^N$ is a finite family of contractions on a complete metric space $(X,d_X)$. Thus, for all $x,y \in X$ and each $j=1,\ldots,N$ we have 
\begin{align}
d_X(\Phi_j(x),\Phi_j(y))	&\leq r_jd_X(x,y),
\end{align}
where $0<r_j<1$ is the \emph{scaling ratio} (or Lipschitz constant) of $\Phi_j$ for each $j=1,\ldots,N$. 

The \emph{attractor} of $\bfPhi$ is the nonempty compact set $F \subset X$ defined as the unique fixed point of the contraction mapping 
\begin{align}
\label{eqn:IFSAsMap}
\displaystyle
\bfPhi(\cdot)	&:= \bigcup_{j=1}^N \Phi_j(\cdot)
\end{align}
on the space of compact subsets of $X$ equipped with the Hausdorff metric. That is, $F=\bfPhi(F)$. If 
\begin{align}
\label{eqn:IFSSelfSimilar}
d_X(\Phi_j(x),\Phi_j(y)) &= r_jd_X(x,y)
\end{align}
for each $j=1,\ldots,N$ (i.e., if the contraction maps $\Phi_j$ are similarities with scaling ratios $r_j$), then the attractor $F$ is the \emph{self-similar set} associated with $\bfPhi$.
\end{definition}

\begin{remark}
\label{rmk:AttentionOnEuclideanSpacesAndOSC}
We focus our attention on Euclidean spaces of the form $X=\R^m$, where $m$ is a positive integer and $d_X=d_m$ is the classic $m$-dimensional Euclidean distance. We  denote $d_m(x,y)$ by $|x-y|$. Furthermore, we consider only iterated functions systems which satisfy the \emph{open set condition} (see \cite{Falc,Hut}). Recall that an IFS $\bfPhi$ satisfies the open set condition if there is a nonempty open set $V \subset \R^m$ for which $\Phi_j(V)\subset V$ for each $j$ and the images $\Phi_j(V)$ are pairwise disjoint.
\end{remark}

%\begin{figure}
%\includegraphics{filename}
%\caption{text of caption}
%\label{}
%\end{figure}

\begin{figure}
\includegraphics[scale=.5]{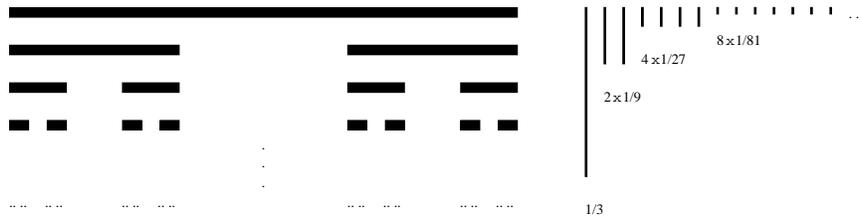}
\caption{The classic ``middle-third removal'' construction of the Cantor set $C$ is depicted on the left. The Cantor string $\calL_{CS}$ is the nonincreasing sequence comprising the lengths of the removed intervals which are depicted on the right as a fractal harp.}
\label{fig:CantorConstructionAndHarp}
\end{figure}

\begin{example}[The dimension of the Cantor set]
The Cantor set $C$ can be constructed in various ways. For instance, we have the classic ``middle-third removal'' construction of $C$ as depicted in Figure 1. 
%\ref{fig:CantorConstructionAndHarp}. 
A more elegant construction shows $C$ to be the unique nonempty attractor of the iterated function system $\bfPhi_C$ on $[0,1]$ given by the two contracting similarities $\varphi_1(x)=x/3$ and $\varphi_2(x)=x/3+2/3$. The box-counting dimension of $C$ is $\log_{3}2$, a fact which can be established with any of the myriad of formulas presented in this paper. Notably, $\log_{3}2$ is equivalently found to be: the \emph{order of the geometric counting function} (see Remark \ref{rmk:DimensionEqualsOrder}) of the \emph{box-counting fractal string} $\calL_B$ of $C$ (which is related but not equal to the Cantor string $\calL_{CS}$, see Definition \ref{def:BoxCountingFractalString}, Equation \eqref{eqn:CantorStringLengths}, and \cite[Ch.1]{LapvF6}); the \emph{abscissa of convergence} of either the \emph{geometric zeta function} of $\calL_{CS}$ (Definition \ref{def:GeometricZetaFunction}), the \emph{box-counting zeta function} of $C$ (Definition \ref{def:BoxCountingZetaFunction}), the \emph{distance zeta function} of $C$ (Definition \ref{def:DistanceZetaFunction}), or the \emph{tube zeta function} of $C$ (Equation \eqref{tube_zeta} in Section \ref{sec:TubeZetaFunction}); or else the unique real-valued solution of the corresponding Moran equation (cf. Equation \eqref{eqn:Moran}):  $2\cdot3^{-s}=1$.
\end{example}

\subsection{Similarity dimension}
\label{sec:SimilarityDimension}

The first notion of `dimension' we consider is the similarity dimension (or `similarity exponent') of a self-similar set.

\begin{definition}
\label{def:SimilarityDimension}
Let $\bfPhi$ be an iterated function system that satisfies the open set condition and is generated by similarities with scaling ratios $\{r_j\}_{j=1}^N$, with $N \geq 2$. Then the \emph{similarity dimension} of the attractor of $\bfPhi$ (that is, of the self-similar set associated with $\bfPhi$) is the unique real solution $D_\bfPhi$ of the equation  
\begin{align}
\label{eqn:Moran}
\sum_{j=1}^N r_j^\sigma	&= 1, \quad \sigma \in \R.
\end{align}
\end{definition}

\begin{remark}
\label{rmk:MoransTheorem}
Equation \eqref{eqn:Moran} is known as Moran's equation. Moran's Theorem is a well-known result which states that the similarity dimension $D_\bfPhi$ is equal to the box-counting (and Hausdorff) dimension of the self-similar attractor of $\bfPhi$.\footnote{Moran's original result in \cite{Mor} was established in $\R$ (i.e., for $m=1$) but is valid for $m\geq 1$; cf. \cite{Falc,Hut}.} In fact, $D_\bfPhi$ is positive, a fact that can be verified directly from Equation \eqref{eqn:Moran}. For details regarding iterated functions systems, the open set condition, and Moran's Theorem, see \cite[Ch.9]{Falc} as well as \cite{Hut} and \cite{Mor}.
\end{remark}

\subsection{Box-counting dimension}
\label{sec:BoxCountingDimension}

In this section we discuss the central notion of box-counting dimension and some of its properties.

\begin{definition}
\label{def:BoxCountingFunction}
Let $A$ be a subset of $\R^m$. The \emph{box-counting function} of $A$ is the function $N_B(A,\cdot): (0,\infty) \rightarrow \N \cup \{0\}$, where (for a given $x>0$) $N_B(A,x)$ denotes the maximum number of disjoint closed balls $B(a,x^{-1})$ with centers $a \in A$ of radius $x^{-1}$.
\end{definition}

\begin{definition}
\label{def:BoxCountingDimension}
For a set $A \subset \R^m$, the \emph{lower} and \emph{upper box-counting dimensions} of $A$, denoted $\underline{\dim}_{B}A$ and $\overline{\dim}_{B}A$, respectively, are given by
\begin{equation}
\begin{aligned}
\underline{\dim}_{B}A&:= \liminf_{x\to\infty}\frac{\log N_B(A,x)}{\log x}, \label{eqn:UpperBoxDimension}\\
\overline{\dim}_{B}A&:= \limsup_{x\to\infty}\frac{\log N_B(A,x)}{\log x}. 
\end{aligned}
\end{equation}

%\begin{equation}
%\begin{aligned}
%\underline{\dim}_{B}A&:= \liminf_{\varepsilon \rightarrow 0^+}\frac{\log N_B(A,\varepsilon^{-1})}{\log \varepsilon^{-1}}, \label{eqn:UpperBoxDimension}\\
%\overline{\dim}_{B}A&:= \limsup_{\varepsilon \rightarrow 0^+}\frac{\log N_B(A,\varepsilon^{-1})}{\log \varepsilon^{-1}}. 
%\end{aligned}
%\end{equation}

%\begin{align}
%\underline{\dim}_{B}A&:= \liminf_{\varepsilon \rightarrow 0^+}\frac{\log N_B(A,\varepsilon^{-1})}{\log\varepsilon^{-1}}, \label{eqn:UpperBoxDimension}\\
%\overline{\dim}_{B}A&:= \limsup_{\varepsilon \rightarrow 0^+}\frac{\log N_B(A,\varepsilon^{-1})}{\log\varepsilon^{-1}}. \notag
%\end{align}

When $\underline{\dim}_{B}A=\overline{\dim}_{B}A$, the following limit exists and is called the \emph{box-counting dimension} of $A$, denoted $\dim_{B}A$:
\begin{align*}
\dim_{B}A	&:= \lim_{x\to\infty}\frac{\log N_B(A,x)}{\log x}.
\end{align*}
\end{definition}

\begin{remark}
\label{rmk:EquivalentUpperBoxDimension}
The upper and lower box-counting dimensions are independent of the ambient dimension $m$. In most applications the set $A$ is such that $N_B(A,x) \asymp x^d$ as $x\to\infty$, for some constant $d \in [0,m]$ (the relation $\asymp$ is explained at the end of Notation \ref{not:DistanceVolumeAndBigO} below). It is easy to see that then, $\dim_{B}A=d$. Also, the following equivalent form of the upper box-counting dimension will prove to be useful in Section \ref{sec:BoxCountingFractalStringsAndZetaFunctions} (see Notation \ref{not:DistanceVolumeAndBigO} below for an explanation of `Big-$O$' notation: $f(x)=O(g(x))$ as $x \rightarrow \infty$).
\begin{align*}
%\label{eqn:EquivalentUpperMinkoskiDimension}
\overline{\dim}_{B}A	&=\inf\{\alpha \geq 0 : N_B(A,x)=O(x^{\alpha}) \textnormal{ as } x \rightarrow \infty\}.
\end{align*}
\end{remark}

\begin{remark}
\label{rmk:VariousCountingFunctions} 
There are many equivalent definitions of the box-counting dimension (see \cite[Ch.3]{Falc}). For instance, the box-counting function $N_B(A,x)$ given in Definition \ref{def:BoxCountingFunction} may be replaced by:
\begin{enumerate}
\setlength{\itemsep}{0in}
\item the minimum number of sets of diameter at most $x^{-1}$ required to cover $A$;
\item the minimum number of closed balls of radius $x^{-1}$ required to cover $A$;
\item the minimum number of closed cubes with side length $x^{-1}$ required to cover $A$; or
\item the number of $x^{-1}$-mesh cubes that intersect $A$.
\end{enumerate}
\end{remark}

\begin{remark}
\label{rmk:varepsilonVSx}
One may also define the box-counting function in terms of $\varepsilon>0$, where $\varepsilon=x^{-1}$ plays the role of the scale under consideration.\footnote{Indeed, note that given $\varepsilon>0$, $N_B(A,\varepsilon^{-1})$ is the maximum number of disjoint balls $B(a,\varepsilon)$ with center $a \in A$ and radius $\varepsilon$ (or, \emph{mutatis mutandis}, $\varepsilon$ denotes any of the counterparts to the notion of radius given in (i)--(iv) of Remark \ref{rmk:VariousCountingFunctions}).} Although this may be a more natural way to describe a box-counting function, the results relating box-counting functions and \emph{geometric counting functions} (see Definition \ref{def:GeometricCountingFunction}) presented in Section \ref{sec:BoxCountingFractalStringsAndZetaFunctions} are stated and analyzed in terms of $x>0$. Moreover, this convention is used throughout \cite{LapvF6} and that text will be vital to the development of further material based on the results presented in this paper.
\end{remark}

\begin{remark}
\label{rmk:PropertiesBoxCountingDimension}
If the box-counting function $N_B(A,x)$ is given as in Definition \ref{def:BoxCountingFunction} or one of the alternatives in Remark \ref{rmk:VariousCountingFunctions}, then the upper and lower box-counting dimensions have the following properties (cf. \cite[Ch.~3]{Falc} and \cite{Fdrr}):
\begin{enumerate}
\setlength{\itemsep}{0in}
\item Let $V$ be a bounded $n$-dimensional submanifold of $\R^m$ which is \emph{rectifiable} in the sense that $V \subset f(\R^n)$, where $f:\R^n\to\R^m$ is a Lipschitz function. Then $\dim_{B}V=n$.
\item Both $\overline{\dim}_B$ and $\underline{\dim}_B$ are monotonic. That is, if $A_1 \subset A_2 \subset \R^m$, then
	\begin{align*}
	%\label{eqn:Monotonicity}
	\overline{\dim}_{B}A_1 &\leq \overline{\dim}_{B}A_2, \quad 			 \underline{\dim}_{B}A_1 \leq \underline{\dim}_{B}A_2.
	\end{align*}

\item Let $\overline{A}$ denote the closure of $A$ (i.e., the smallest closed subset of $\R^m$ which contains $A$). Then 
	\begin{align*}
	%\label{eqn:ClosureSameDimension}
	\overline{\dim}_{B}\overline{A}	&= \overline{\dim}_{B}A, \quad
	\underline{\dim}_{B}\overline{A} = \underline{\dim}_{B}A.
	\end{align*}
\item  For any two sets $A_1,A_2 \subset \R^m$,
	\begin{align*}
	%\label{eqn:FiniteStability}
	\overline{\dim}_{B}(A_1 \cup A_2)&= \max\left\{\overline{\dim}_{B}A_1,\overline{\dim}_{B}A_2\right\}.  
	\end{align*} That is, $\overline{\dim}_B$ is finitely stable. On the other hand, $\underline{\dim}_B$ is not finitely stable. 
\item Neither $\overline{\dim}_B$ nor $\underline{\dim}_B$ is countably stable. That is, neither $\overline{\dim}_B$ nor $\underline{\dim}_B$ satisfies the analogue of property (iv) for a countable collection of subsets of $\R^m$. 	
\end{enumerate}
\end{remark}

Concerning the loss of finite stability of the lower box dimension mentioned in
(iv) of Remark \ref{rmk:PropertiesBoxCountingDimension}, it is possible to construct two bounded sets $A$ and $B$ in $\R^m$ such that their lower box-counting dimensions are both equal to zero, while the box-counting dimension of their union is equal to $m$; see \cite[Theorem 1.4]{Zu}.

A simple way to see why property (v) of Remark \ref{rmk:PropertiesBoxCountingDimension} is satisfied for the upper box-counting dimension is to consider the countable set $A=\{1,1/2,1/3,\ldots\}$ and note that $\overline{\dim}_{B}A=1/2$ whereas $\overline{\dim}_{B}\{1/j\}=0$ for each positive integer $j$.

The following proposition shows that one need only consider certain discrete sequences of scales which tend to zero in order to determine the box-counting dimension of a set.

\begin{proposition}
\label{prop:GeometricSequenceSufficient}
Let $\lambda>1$ and $A \subset \R^m$. Then
\begin{align*}
%\label{eqn:GeometricSequenceSufficient}
\underline{\dim}_{B}A	&=\liminf_{k \rightarrow \infty}\frac{\log N_B(A,\lambda^{k})}{\log\lambda^{k}}, \\
\overline{\dim}_{B}A	&=\limsup_{k \rightarrow \infty}\frac{\log N_B(A,\lambda^{k})}{\log\lambda^{k}}.
\end{align*}
\end{proposition}

\begin{proof}
If $\lambda^{k}< x \leq\lambda^{k+1}$, then
\begin{align*}
\frac{\log N_B(A,x)}{\log x} &\leq  
\frac{\log N_B(A,\lambda^{k+1})}{\log\lambda^{k}} = 
\frac{\log N_B(A,\lambda^{k+1})}{\log\lambda^{k+1}-\log\lambda}.
\end{align*}
Therefore, 
\begin{align*}
\limsup_{x\to\infty}\frac{\log N_B(A, x)}{\log  x}	& \leq \limsup_{k \rightarrow \infty}\frac{\log N_B(A,\lambda^{k})}{\log\lambda^{k}}.
\end{align*}
The opposite inequality clearly holds and the case for the lower limits follows \emph{mutatis mutandis}.
\end{proof}

\begin{example}[Box-counting dimension of the Cantor set]
\label{eg:BoxCountingDimensionOfCantorSet}
Let $C$ be the Cantor set and $n \in \N$. Also, let $N_B(A,3^n)$ denote the minimum number of disjoint closed intervals with length $3^{-n}$ required to cover $C$. Then
$N_B(A,3^n)=2^n$, so by Proposition \ref{prop:GeometricSequenceSufficient} we have
\begin{align}
\label{eqn:BoxCountingDimensionOfCantorSet}
\dim_{B}C	&= \lim_{n\rightarrow\infty} \frac{\log{2^n}}{\log{3^n}}=\log_{3}2.
\end{align}
\end{example}

In the next section, we discuss the Minkowski dimension, which is well known to be equivalent to the box-counting dimension.

\subsection{Minkowski dimension}
\label{sec:MinkowskiDimension}

Minkowski content and Minkowski dimension require a specific notion of volume and can be stated concisely with the following notation.

\begin{notation}[Distance, volume, and Big-$O$]
\label{not:DistanceVolumeAndBigO}
Let $\varepsilon >0$ and $A \subset \R^m$. Let $d(x,A)$ denote the distance between a point $x\in\R^m$ and the set $A$ given by
\begin{align*}
%\label{eqn:PointSetDistance}
d(x,A)	&:= \inf\{|x-u|_m : u \in A\},
\end{align*}
where $|\cdot|_m$ denotes the $m$-dimensional Euclidean norm. The \emph{$\varepsilon$-neighborhood} of $A$, denoted $A_\varepsilon$, is the set of points in $\R^m$ which are within $\varepsilon$ of $A$. Specifically,
\begin{align*}
%\label{eqn:DeltaNeighborhood}
A_\varepsilon	&= \{x \in \R^m: d(x,A) < \varepsilon\}.
\end{align*}

In the sequel, we fix the set $A$ and are concerned with the $m$-dimensional Lebesgue measure (denoted $\vol_m$) of its $\varepsilon$-neighborhood $A_\varepsilon$ for a given $\varepsilon>0$. Recall, for completeness,\footnote{The book \cite{Cohn} is a good general reference on elementary (as well as more advanced) measure theory.} that the $m$-dimensional Lebesgue measure of a (measurable) set $A \subset \R^m$ is given by 
\begin{align*}
%\label{eqn:LebesgueMeasure}
\vol_m(A)	&:=\inf\left\{\sum_{n=1}^\infty \prod_{j=1}^m(b_{n,j}-a_{n,j}) : A \subset \bigcup_{n=1}^\infty \left(\prod_{j=1}^m[a_{n,j},b_{n,j}]\right) \right\}.
\end{align*}

In the case of an \emph{ordinary fractal string} $\Omega \subset \R$ (see the latter part of Definition \ref{def:FractalString}), we are interested in the 1-dimensional volume (i.e., length) of the  \emph{inner} $\varepsilon$-neighborhood of the boundary $\partial\Omega$. Specifically, given an ordinary fractal string $\Omega$ and $\varepsilon>0$, the volume $V_{\textnormal{inner}}(\varepsilon)$ of the inner $\varepsilon$-neighborhood of $\partial\Omega$ is defined by 
\begin{align}
\label{eqn:VolumeInnerNeighborhood}
V_{\textnormal{inner}}(\varepsilon)	&:=\vol_1\{x \in \Omega : d(x,\partial\Omega) <\varepsilon\}.
\end{align}

For two functions $f$ and $g$, with $g$ nonnegative, we write $f(x)=O(g(x))$ as $x\to\infty$ if there exists a positive real number $c$ such that for all sufficiently large $x$, $|f(x)|\leq cg(x)$. More generally, if there exists $C$ such that $|f(x)|\leq Cg(x)$ for all $x$ sufficiently close to some value $a \in \R\cup \{\pm\infty\}$, then we write $f(x)=O(g(x))$ as $x\to a$. If both $f(x)=O(g(x))$ and $g(x)=O(|f(x)|)$ as $x\to a$, we write $f(x)\asymp g(x)$ as $x\to a$. Moreover, if $\lim_{x\to a} f(x)/g(x)=1$,\footnote{Or, more generally, if $f(x)=g(x)(1+o(1))$ as $x\to a$, where $o(1)$ stands for a function tending to zero as $x\to a$.} then we write $f(x) \sim g(x)$ as $x\to a$. Analogous notation will be used for infinite sequences.
\end{notation}

\begin{definition}[Minkowski content]
\label{def:MinkowskiContent0}
Let $r$ be a given nonnegative real number. The \emph{upper} and \emph{lower $r$-dimensional Minkowski contents} of a bounded set $A\subset \R^m$ are respectively given by
\begin{align*}
\mink^{*r}(A) &:= \limsup_{\varepsilon\rightarrow 0^+}\frac{\vol_m(A_\varepsilon)}{\varepsilon^{m-r}},\\
\mink_*^{r}(A) &:= \liminf_{\varepsilon\rightarrow 0^+}\frac{\vol_m(A_\varepsilon)}{\varepsilon^{m-r}}.
\end{align*}
\end{definition}

It is easy to see that if $\mink^{*r}(A)<\infty$, then $\mink^{*s}(A)=0$ for each $s>r$. Furthermore, since $A$ is bounded, then clearly $\mink^{*r}(A)=0$ for $r>m$. On the other hand, if $\mink^{*r}(A)>0$, then $\mink^{*s}(A)=\infty$ for each $s<r$. Therefore, there exists a unique point in $[0,m]$ at which the function $r\mapsto \mink^{*r}(A)$ jumps from the value of $\infty$ to zero. This unique point is called the \emph{upper Minkowski dimension} of $A$. The \emph{lower Minkowski dimension} of $A$ is defined analogously by using the lower $r$-dimensional Minkowski content.

\begin{definition}[Minkowski dimension]
\label{def:MinkowskiDimension0}
The \emph{upper} and \emph{lower Minkowski dimensions} of a bounded set $A$  are defined respectively by
\begin{equation}
\begin{aligned}
\overline{\dim}_{M}A 	&:=\inf\{r\ge0:\mink^{r*}(A)=0\}=\sup\{r\ge0:\mink^{r*}(A)=\infty\}, \label{eqn:UpperMinkowskiDimension}\\
\underline{\dim}_{M}A 	&:= \inf\{r\ge0:\mink_*^{r}(A)=0\}=\sup\{r\ge0:\mink_*^{r}(A)=\infty\}.
\end{aligned}
\end{equation}
When $\overline{\dim}_{M}A=\underline{\dim}_{M}A$, the common value is called the \emph{Minkowski dimension} of $A$, denoted by $\dim_{M}A$.
\end{definition}

When we write $\dim_{M}A$, we implicitly assume that the Minkowski dimension of $A$  exists. In most applications we have that $\vol_m(A_\varepsilon)\asymp\varepsilon^{\alpha}$ as $\varepsilon\to0^+$, where $\alpha$ is a number in $[0,m]$. Then $\dim_{M}A$ exists and is equal to $m-\alpha$ (in light of Definitions \ref{def:MinkowskiContent0} and \ref{def:MinkowskiDimension0}). Note that here 
\begin{equation*}
\alpha=\lim_{\varepsilon\to0^+}\frac{\log\vol_m(A_\varepsilon)}{\log\varepsilon},
\end{equation*}
and hence, 
\begin{equation*}
\dim_{M}A=m-\lim_{\varepsilon\to0^+}\frac{\log\vol_m(A_\varepsilon)}{\log\varepsilon}.
\end{equation*}
It is not difficult to show that the following more general result holds.

\begin{proposition}
\label{prop:MinkowskiDimension}
The upper and lower Minkowski dimensions of a bounded set $A\subset \R^m$ are respectively given by
\begin{align*}
\overline{\dim}_{M}A &= m-\liminf_{\varepsilon \rightarrow 0^+}\frac{\log\vol_m(A_\varepsilon)}{\log\varepsilon},\\
\underline{\dim}_{M}A &= m-\limsup_{\varepsilon \rightarrow 0^+}\frac{\log\vol_m(A_\varepsilon)}{\log\varepsilon}.
\end{align*}
\end{proposition}

\begin{remark}
\label{rmk:EquivalentUpperMinkoskiDimension}
The upper and lower Minkowski dimensions are, of course, independent of the ambient dimension $m$. The upper Minkowski dimension is equivalently given by
\begin{align*}
%\label{eqn:EquivalentUpperMinkoskiDimension}
\overline{\dim}_{M}A	&=\inf\{\alpha \geq 0 : \vol_m(A_\varepsilon)=O(\varepsilon^{m-\alpha}) \textnormal{ as } \varepsilon \rightarrow 0^+\}.
\end{align*}
%This equivalent form of the upper Minkowski dimension will prove to be useful in Section \ref{sec:BoxCountingFractalStringsAndZetaFunctions}.
\end{remark}

\begin{remark}
\label{rmk:Tr2}
It is interesting that there exists a bounded set $A$ in $\R^m$ such that the upper and lower box dimension are different (see, e.g., \cite[p.\ 122]{Tr2}), and even such that $\overline{\dim}_{M}A=m$ and $\underline{\dim}_{M}A=0$ (see \cite[Theorem 1.2]{Zu}).
\end{remark}

\begin{remark}
\label{rmk:HaPo}
The upper Minkowski dimension of $A$ is important in the study of the Lebesgue integrability of the distance function $d(x,A)^{-\gamma}$ in an $\varepsilon$-neighborhood of $A$, where $\varepsilon$ is a fixed positive number:
\begin{equation}
\mbox{If $\gamma<m-\overline\dim_{M}A$, then $\int_{A_{\varepsilon}}d(x,A)^{-\gamma}dx<\infty$.}
\end{equation} 
This nice result is due to Harvey and Polking, and is implicitly stated in \cite{HaPo}; see also \cite{Zu} for related results and references. This fact enabled the first and third authors, along with G.~Radunovi$\acute{\textnormal{c}}$, to determine the abscissa of convergence of the so-called distance zeta function of $A$; see Definition~\ref{def:DistanceZetaFunction} below along with Theorem \ref{thm:DistanceZetaFunctionUpperBoxDimension} and \cite{LapRaZu} for details.
\end{remark}

\begin{definition}[Minkowski measurability]
\label{def:MinkowskiMeasurability}
Let $A \subset \R^m$ be such that $D_M=\dim_{M}A$ exists. The \emph{upper} and \emph{lower Minkowski content} of $A$ are respectively defined as its $D_M$-dimensional upper 
and lower Minkowski contents, that is,
\begin{align*}
\mink^* &:=\mink^{*D_M}(A)=\limsup_{\varepsilon\to 0^+} \frac{\vol_m(A_\varepsilon)}{\varepsilon^{m-D_M}},\\
\mink_* &:=\mink_*^{D_M}(A)=\liminf_{\varepsilon\to 0^+} \frac{\vol_m(A_\varepsilon)}{\varepsilon^{m-D_M}}.
\end{align*} 
If the upper and lower Minkowski contents agree and lie in $(0,\infty)$, then $A$ is said to be \emph{Minkowski measurable} and the \emph{Minkowski content of $A$} is given by
\begin{align*}
\mink &:=\lim_{\varepsilon\to 0^+}\frac{\vol_m(A_\varepsilon)}{\varepsilon^{m-D_M}}.
\end{align*} 
\end{definition}

For example, if $A$ is such that $\vol_m(A_\varepsilon) \sim M\varepsilon^{\alpha}$ as $\varepsilon \to 0^+$, then $\dim_{B}A=m-\alpha$ and $\mink=M$.

\begin{openproblem}
If $A$ and $B$ are Minkowski measurable in
$\mathbf{R}^m$ and $\mathbf{R}^n$, respectively, is their Cartesian product $A\times B$ Minkowski measurable in $\mathbf{R}^{m+n}$\emph{?} See also Remark \ref{rmk:Resman} below dealing with the so-called normalized Minkowski content, and its independence of the ambient dimension~$m$. 
\end{openproblem}

\begin{remark}
Another question to consider is whether or not the union $A\cup B$ of two Minkowski measurable sets is Minkowski measurable. If not, it would be interesting to find an explicit counter-example. (The answer is clearly affirmative if $A$ and $B$ are a positive distance apart.)
\end{remark}

The Minkowski measurable sets on the real line have been characterized in \cite{LapPo1}; see also Theorem~\ref{thm:CriterionForMinkowskiMeasurability}. Some classes of Minkowski measurable sets are known in the plane in the case of smooth spirals, see \cite{ZuZup}, and in the case of discrete spirals, see \cite{MeSa}. It is interesting that in general, bilipschitz $C^1$ mappings do not preserve Minkowski measurability, even for subsets of the real line; see \cite{MeSa}.

We close this section with the following example.

%\begin{example}
%Let $A$ be a bounded, Lebesgue non-measurable set in $\R^m$. Then $\dim_{B}A=m$. Indeed, the closure $\overline{A}$ cannot be of Lebesgue measure zero (i.e., we cannot have $\vol_m(\overline{A})=0$) since, in that case, $A$ would also be of Lebesgue measure zero, implying that $A$ is Lebesgue measurable. But then $\vol_m(\overline{A})>0$ immediately implies that $\underline{\dim}_{B}\overline{A}=m$, and therefore $\underline{\dim}_{B}A=m$. Since  $\overline{\dim}_{B}A \leq m$, this proves that $\dim_{B}A$ exists and $\dim_{B}A=m$.
%\end{example} 

\begin{example}
If $A$ is any bounded set in $\R^m$ such that its closure is of positive $m$-dimensional Lebesgue measure, then $\dim_BA = m$. Indeed, $\vol_m(\overline{A}) > 0$ immediately implies that $\underline{\dim}_B\overline{A} = m$, and therefore by property (iii) of Remark \ref{rmk:PropertiesBoxCountingDimension}, $\underline{\dim}_BA = m$. Since $\overline{\dim}_BA \leq m$, this proves that $\dim_BA$ exists and $\dim_BA = m$. In particular, the claim holds for any Lebesgue nonmeasurable set $A\subset\R^m$. Indeed, the closure of such a set $A$ cannot be of Lebesgue measure zero (i.e., we cannot have $\vol_m(\overline{A}) = 0$) since, in that case, $A$ would also be of Lebesgue measure zero, implying that $A$ is Lebesgue measurable (because Lebesgue measure is complete in this setting).\footnote{For the notions of measure theory used in this example, we refer, e.g., to \cite{Cohn}.}
\end{example}

\section{Fractal strings and zeta functions}
\label{sec:FractalStringsAndZetaFunctions}

In this section, we discuss a few of the many results on fractal strings presented in \cite{LapvF6}.

\subsection{Fractal strings and ordinary fractal strings}
\label{sec:FractalStringsAndOrdinaryFractalStrings}

\begin{definition}
\label{def:FractalString} 
A \emph{fractal string} $\calL$ is a nonincreasing sequence of positive real numbers which tends to zero. Hence,
\begin{align*}
%\label{eqn:FractalString}
\calL &= (\ell_j)_{j \in \N}, 
\end{align*}
where $(\ell_j)_{j\in\N}$ is nonincreasing and $\lim_{j \rightarrow \infty}\ell_j=0$. Equivalently, a fractal string $\calL$ can be thought of as the collection 
\[
\{ l_n : l_n \textnormal{ has multiplicity } m_n, n \in \N \},
\] 
where $(l_n)_{n \in \N}$ is a strictly decreasing sequence of positive real numbers and, for each $n\in\N$, $m_n$ is the number of lengths (or, more generally, `scales') $\ell_j$ such that $\ell_j=l_n$. 

An \emph{ordinary fractal string} $\Omega$ is a bounded open subset of the real line. In that case, we can write $\Omega$ as the disjoint union of its connected components $I_j$ (i.e., $\Omega=\cup_{j=1}^\infty I_j$), and for each $j\geq 1$, $\ell_j$ denotes the length of the interval $I_j$ (i.e., $\ell_j=|I_j|_1$).\footnote{Note that without loss of generality, we may assume that $\ell_1\geq\ell_2\geq\ldots$, with $\ell_j\to 0$ as $j\to \infty$. We ignore the trivial case where $\Omega$ is a finite union of open intervals; see Remark \ref{rmk:FiniteFractalStrings}.}  
\end{definition}

\begin{remark}
\label{rmk:FiniteFractalStrings}
In \cite{LapvF6}, for instance, finite fractal strings (i.e., nonincreasing sequences of real numbers with a finite number of positive terms) are allowed. However, for reasons described in Remark \ref{rmk:AbscissaForFiniteSequences}, the finite case is not considered in this paper.
\end{remark}

If, as in Definition \ref{def:FractalString} above, an ordinary fractal string $\Omega$ is written as the union of a countably infinite collection of disjoint open intervals $I_j$ (necessarily its connected components), then the lengths $\ell_j$ of the intervals $I_j$ comprise a fractal string $\calL$. Moreover, $\dim_M(\partial\Omega)$ is given by
\begin{align} 
\label{eqn:OrdinaryFractalStringMinkowski} \dim_M(\partial\Omega)	&= \inf\{\alpha\geq 0 : V_{\textnormal{inner}}(\varepsilon)= O(\varepsilon^{1-\alpha}) \textnormal{ as } \varepsilon \to 0^+\},
\end{align}
where $V_{\textnormal{inner}}(\varepsilon)$ is the $1$-dimensional Lebesgue measure of the inner $\varepsilon$-neighborhood of $\Omega$ (see formula  \eqref{eqn:VolumeInnerNeighborhood} in Notation \ref{not:DistanceVolumeAndBigO}). In fact, Equation \eqref{eqn:OrdinaryFractalStringMinkowski} is used to define the \emph{Minkowski dimension} of (the boundary of) an ordinary fractal string in \cite{LapvF6}.\footnote{More specifically, $\dim_M(\partial\Omega)$ should really be denoted by $\dim_{M,\textnormal{inner}}(\partial\Omega)$ and called the \emph{inner Minkowski dimension} of $\partial\Omega$ (or of $\calL$).} Moreover, it is shown in \cite{LapPo1} that $V_{\textnormal{inner}}(\varepsilon)$, and hence also $\dim_M(\partial\Omega)$, depends only on the fractal string $\calL$ (but not on the particular rearrangement of the intervals $I_j$ composing $\Omega$).

\begin{definition}
\label{def:AbscissaOfConvergence}
Let $\calL$ be a fractal string. The \emph{abscissa of convergence} of the Dirichlet series $\sum_{j=1}^\infty\ell_j^s$ is defined by
\begin{align}
\label{eqn:AbscissaOfConvergence}
\sigma	&=\inf \left\{ \alpha \in \R : \sum_{j=1}^\infty\ell_j^\alpha<\infty \right\}.
\end{align}
Thus, $\{ s \in \C : \textnormal{Re}(s)>\sigma \}$ is the largest open half-plane on which this series converges; see, e.g., \cite[\S VI.2]{Ser}.
\end{definition}

\begin{remark}
\label{rmk:AbscissaForFiniteSequences}
If $\calL$ were allowed to be a finite sequence of positive real numbers (as in \cite{LapvF6}), then we would have $\sigma=-\infty$ since the corresponding Dirichlet series would be an entire function. In the context of this paper, we always have that $\sigma\geq 0$ (since $\sum_{j=1}^\infty \ell_j^\alpha$ is clearly divergent when $\alpha=0$). This explains why we consider only (bounded) infinite sets in the development of \emph{box-counting fractal strings} in Section \ref{sec:BoxCountingFractalStringsAndZetaFunctions}. Indeed, for clarity of exposition, we only consider fractal strings consisting of infinitely many positive lengths (or scales), and hence, ordinary fractal strings comprising infinitely many disjoint intervals. 
\end{remark}

\begin{remark}
\label{rmk:ConvergenceAndDivergence}
A key distinction between a fractal string $\calL$ and an ordinary fractal string $\Omega$ lies in the sum of the corresponding lengths (or scales), denoted $(\ell_j)_{j\in\N}$ in either case. Specifically, since an ordinary fractal string $\Omega$ is bounded, $\sum_{j=1}^\infty \ell_j$ is necessarily convergent. On the other hand, for a fractal string $\calL$, $\sum_{j=1}^\infty \ell_j$ may be divergent. See Example \ref{eg:BoxCountingFractalStringOf1DFractal} for a bounded set in $\R^2$ whose \emph{box-counting fractal string} is a fractal string whose lengths have an unbounded sum and yet contains pertinent information regarding the bounded set. (In a somewhat different setting, many other classes of examples are provided in \cite[esp. \S 13.1 \& \S 13.3]{LapvF6} and in \cite{LapPe2,LapPeWin}.)
\end{remark}

\begin{definition}
\label{def:GeometricZetaFunction}
Let $\calL$ be a fractal string. The \emph{geometric zeta function} of $\calL$ is defined by 
\begin{align}
\label{eqn:GeometricZetaFunction}
\zeta_\calL(s)	&= \sum_{j=1}^\infty\ell_j^s,
\end{align}
where $s \in \C$ and  $\textnormal{Re}(s)>D_\calL:=\sigma$. The \emph{dimension} of $\calL$, denoted $D_\calL$, is defined as the abscissa of convergence $\sigma$ of the Dirichlet series which defines $\zeta_\calL$.
\end{definition}

In order to define the complex dimensions of a fractal string, as in \cite{LapvF6}, we assume there exists a meromorphic extension of the geometric zeta function $\zeta_\calL$ to a suitable region. First, consider the \emph{screen} $S$ as the contour 
\begin{align}
\label{eqn:Screen}
S:S(t)+it \quad &(t \in \R),
\end{align}
where $S(t)$ is a continuous function $S:\R \to [-\infty,D_\calL]$. Next, consider the \emph{window} $W$ as the set 
\begin{align}
\label{eqn:Window}
W	&= \{s \in \C : \textnormal{Re}(s) \geq S(\textnormal{Im}(s))\}.
\end{align}

By a mild abuse of notation, we denote by $\zeta_\calL$ both the geometric zeta function of $\calL$ and its meromorphic extension to some region. 

\begin{definition}
\label{def:ComplexDimensions}
Let $W \subset \C$ be a window on an open connected neighborhood of which $\zeta_\calL$ has a meromorphic extension. The set of (\emph{visible}) \emph{complex dimensions} of $\calL$ is the set $\calD_\calL=\calD_\calL(W)$ given by
\begin{align}
\calD_\calL	&= \left\{ \omega \in W : \zeta_\calL \textnormal{ has a pole at } \omega \right\}.
\end{align}
\end{definition}

\noindent In the case where $\zeta_\calL$ has a meromorphic extension to $W=\C$, the set $\calD_\calL$ is referred to as the \emph{complex dimensions} of $\calL$.  Such is the case for the Cantor string $\Omega_{CS}$.

\begin{example}[Complex dimensions of the Cantor string]
\label{eg:ComplexDimensionsOfCantorString}
The Cantor string $\Omega_{CS}$ is the ordinary fractal string given by $\Omega_{CS}=[0,1] \setminus C$, where $C$ is the Cantor set (see Example \ref{eg:BoxCountingDimensionOfCantorSet}). The lengths of the Cantor string are given by the fractal string 
\begin{align}
\label{eqn:CantorStringLengths}
\calL_{CS}	&=\{3^{-n}: 3^{-n} \textnormal{ has multiplicity } 2^{n-1}, n \in \N\}.
\end{align}
The geometric zeta function of the Cantor string, denoted $\zeta_{CS}$, is given by
\begin{align}
\label{eqn:CantorStringGeometricZetaFunction}
\zeta_{CS}(s)	&:= \zeta_{\calL_{CS}}(s)= \sum_{n=1}^\infty 2^{n-1}3^{-ns} = \frac{3^{-s}}{1-2\cdot3^{-s}}.
\end{align} 
The closed form on the right-hand side of Equation \eqref{eqn:CantorStringGeometricZetaFunction} allows for the meromorphic continuation of $\zeta_{CS}$ to all of $\C$. Hence, $\zeta_{CS}=3^{-s}(1-2\cdot 3^{-s})^{-1}$ for all $s \in \C$. It follows that the complex dimensions of the Cantor string, denoted $\calD_{CS}$, are the complex roots of the Moran equation $2\cdot3^{-s}=1$. Thus, with the window $W$ chosen to be all of $\C$ (so that $\calD_{CS}:=\calD_{\calL_{CS}}=\calD_{\calL_{CS}}(\C)$, in the notation of Definition \ref{def:ComplexDimensions}), we have
\begin{align}
\label{eqn:CantorStringComplexDimensions}
\calD_{CS}	&=\left\{ \log_3{2} + i\frac{2\pi}{\log 3}z : z \in \Z \right\}.
\end{align}
Note that by Equation \eqref{eqn:BoxCountingDimensionOfCantorSet} the dimension $D_{CS}:=D_{\calL_{CS}}=\log_{3}2$ of the Cantor string coincides with $\dim_{B}C=\dim_{M}C$, and this value is the unique real-valued complex dimension of the Cantor string.
\end{example}

\subsection{Geometric counting function of a fractal string}
\label{sec:GeometricCountingFunctionOfAFractalString}

The results in this section connect the counting function of the lengths of a fractal string to its dimension and geometric zeta function.

\begin{definition}
\label{def:GeometricCountingFunction}
The \emph{geometric counting function} of $\calL$, or the \emph{counting function of the reciprocal lengths} of $\calL$, is given by
\begin{align*}
%\label{eqn:GeometricCountingFunction}
\displaystyle N_\calL(x) &:= \#\{j \in \N : \ell_j^{-1} \leq x\} = \sum_{n \in \N, \, l_n^{-1} \leq\, x}m_n
\end{align*}
for $x>0$.
\end{definition}

The following easy proposition is identical to Proposition 1.1 of \cite{LapvF6}.

\begin{proposition}
\label{prop:CountingFunctionD}
Let $\alpha \geq 0$ and $\calL$ be a fractal string. Then $N_\calL(x)=O(x^\alpha)$ as $x \rightarrow \infty$ if and only if
$\ell_j=O(j^{-1/\alpha})$ as $j \rightarrow \infty$. 
\end{proposition}

\begin{proof}
Suppose that for some $C>0$ we have 
\[
N_\calL(x) \leq Cx^\alpha.
\]
Let $x=\ell_j^{-1}$, then $j \leq C\ell_j^{-\alpha}$, which implies that 
\[
\ell_j = O(j^{-1/\alpha}).
\] 
Conversely, if $\ell_j \leq cj^{-1/\alpha}$
for $j \in \N$ and some $c>0$, then given $x>0$, we have 
\[
\ell_j^{-1} > x \textnormal{ for } j > (cx)^\alpha.
\]
Therefore, 
\[
N_\calL(x) \leq (cx)^\alpha.
\]
\end{proof}

\begin{remark}
\label{rmk:MoreAsymptoticBehavior}
Many additional (and harder) results connecting the asymptotic behavior of the geometric counting function, the spectral counting function, and the (upper and lower) Minkowski content(s) of a fractal string $\calL$ are provided in \cite{LapPo1}. The simplest one states that $N_\calL(x)=O(x^\alpha)$ as $x\to\infty$ (i.e., $\ell_j=O(j^{-1/\alpha})$ as $j\to\infty$) if and only if $\mink^{*\alpha}(\partial\Omega)<\infty$, where (consistent with our earlier comment) $\mink^{*\alpha}(\partial\Omega)$ is given as in Definition \ref{def:MinkowskiContent0} except with $\vol_m(\cdot)$ replaced with $V_{\textnormal{inner}}(\cdot)$.
\end{remark}

\begin{notation}
\label{not:D}
The infimum of the nonnegative values of $\alpha$ which satisfy Proposition \ref{prop:CountingFunctionD} plays a key role in our results. Hence, we let $D_N$ denote that special value. That is,
\begin{align}
\label{eqn:CountingFunctionD}
D_N:= \inf\{\alpha \geq 0 : N_\calL(x)=O(x^\alpha) \textnormal{ as } x\to\infty\}.
\end{align}
\end{notation}

The following lemma is a restatement of a portion of Lemma 13.110 of \cite{LapvF6}.\footnote{In fact, a stronger result holds in the setting of generalized fractal strings (viewed as measures) in \cite{LapvF6}, but it is beyond the scope of this paper. (See also \cite{LapLuvF1}.)}

\begin{lemma}
\label{lem:GeometricZetaFunctionIntegralTransformCountingFunction}
Let $\calL$ be a fractal string. Then 
\begin{align}
\label{eqn:GeometricZetaFunctionIntegralTransformCountingFunction}
\zeta_\calL(s)	&= s\int_0^\infty N_\calL(x)x^{-s-1}dx
\end{align}
and, moreover, the integral converges \emph{(}and hence, Equation \eqref{eqn:GeometricZetaFunctionIntegralTransformCountingFunction} holds\emph{)}  if and only if $\sum_{j=1}^\infty \ell_j^s$ converges, i.e., if and only if $\textnormal{Re}(s)>D_\calL=\sigma$.
\end{lemma}

\begin{proof}
Let $s \geq 0$ be a real number. (The case where $s \in \C$ follows immediately from this case by analytic continuation for $\textnormal{Re}(s)>D_\calL=\sigma$ since $\zeta_\calL$ is holomorphic in that half-plane.) For any given $n \in \N$, we have  
\begin{align*}
s\int_0^{\ell_n^{-1}} N_\calL(x)x^{-s-1}dx	&= \sum_{j=1}^{n-1}s\int_{\ell_j^{-1}}^{\ell_{j+1}^{-1}} N_\calL(x)x^{-s-1}dx
	= \sum_{j=0}^{n-1}j(\ell_j^s-\ell_{j+1}^s)
\end{align*}
since $N_\calL(x)=0$ for $x<\ell_1^{-1}$ and $N_\calL(x)=j$ for $\ell_j^{-1} \leq x < \ell_{j+1}^{-1}$. Furthermore,
\begin{align*}
s\int_0^{\ell_j^{-1}} N_\calL(x)x^{-s-1}dx	&= 
\sum_{j=1}^{n-1}j\ell_j^s-\sum_{j=1}^{n}(j-1)\ell_j^s
	= \sum_{j=1}^{n}\ell_j^s-n\ell_n^s.
\end{align*}
Now, for $s\geq 0$, we have $n\ell_n^s \leq 2\sum_{j=[n/2]}^n\ell_j^s$. Thus, Equation \eqref{eqn:GeometricZetaFunctionIntegralTransformCountingFunction} holds if and only if $\sum_{j=1}^\infty \ell_j^s$ converges.\footnote{That is, if and only if $s>D_\calL=\sigma$ (or, more generally, if $s\in\C$, if and only if $\textnormal{Re}(s)>D_\calL$).} Moreover,
\begin{align*}
\lim_{n \rightarrow \infty} s\int_0^{\ell_n^{-1}} N_\calL(x)x^{-s-1}dx	&= \zeta_\calL(s) 
\end{align*}
since the tail $\sum_{j=[n/2]}^\infty \ell_j^s$ converges to zero. (Here, $[y]$ denotes the integer part of the real number $y$.)
\end{proof}

The following proposition will be used to prove a portion of our main result, Theorem \ref{thm:MainResult} (cf. Theorem 13.111 and Corollary 13.112 of \cite{LapvF6}, as well as  \cite{LapLuvF1,LapLuvF2}, where this proposition is established in the context of $p$-adic fractal strings and also of ordinary (real) fractal strings). %The proof of this proposition stems from the proof of a preliminary version of Theorem \ref{thm:MainResult} which was written by the third author during the First Annual Meeting of PISRS in 2011.

\begin{proposition}
\label{prop:CountingFunctionDFractalStringD}
Let $\calL$ be a fractal string. Then 
\begin{align*}
%\label{eqn:DimensionEqualsCountingFunctionD}
D_\calL	&=D_N, %:= \inf\{\alpha \geq 0 : N_\calL(x)=O(x^\alpha)\}.
\end{align*}
where $D_\calL=\sigma$ is the dimension of $\calL$ given by Equation \eqref{eqn:AbscissaOfConvergence} \emph{(}and Definition \ref{def:GeometricZetaFunction}\emph{)} and $D_N$ is given by Equation \eqref{eqn:CountingFunctionD}.
\end{proposition}

\begin{proof}
The proof given here follows that of \cite{LapvF6}, \emph{loc.~cit.} (See also \cite{LapLuvF1}.) First, suppose $\textnormal{Re}(s)>D_N$. Denoting $t=\textnormal{Re}(s)$, we choose any fixed  $\alpha\in(D_N,t)$.
 Using Lemma \ref{lem:GeometricZetaFunctionIntegralTransformCountingFunction}, for $x_1=(\ell_1)^{-1}$ we have
\begin{align*}
|\zeta_\calL(s)|
	&\leq |s|\int_{x_1}^\infty Cx^{\alpha}x^{-t-1}dx = \left[\frac{|s|Cx^{\alpha-t}}{\alpha-t}\right]_{x_1}^\infty 
	= 0 - \frac{|s|Cx_1^{\alpha-t}}{\alpha-t},
\end{align*}
since $\alpha-t<0$. Hence, $|\zeta_\calL(s)|<\infty$. In other words,  $t>D_{\calL}$ for any $t>D_N$. Letting $t\searrow D_N$, we obtain that $D_{\calL}\le D_N$.

%The converse inequality $D_{\calL}\ge D_N$ follows similarly as in the proof of Theorem~13.111 of \cite{LapvF6}.

For the converse, suppose $\alpha < D_N$. Then $N_\calL(x)$ is not $O(x^\alpha)$ as $x \rightarrow \infty$. So, there exists a strictly increasing sequence $(x_j)_{j\in\N}$ with $x_1 \geq \ell_1^{-1}$ which tends to infinity such that 
\[
N_\calL(x) \geq jx_j^\alpha
\]
for each $j$. Then, for $t \leq 1$, 
\[
t\int_{\ell_1^{-1}}^\infty N_\calL(x)x^{-t-1}dx \geq \sum_{j=1}t\int_{x_j}^{x_{j+1}}jx^\alpha x^{-t-1}dx
\]
since $N_\calL(x)$ is increasing.

We estimate $t\int_{x_j}^{x_{j+1}}x^{-t-1}dx \geq x_j^{-t}$ to obtain
\[
t\int_{\ell_1^{-1}}^\infty N_\calL(x)x^{-t-1}dx \geq \sum_{j=1}^\infty jx_j^{\alpha-t}.
\]
For $t \leq \alpha$, the sum diverges. Hence, $D_\calL \geq \alpha$ for all $\alpha<D_N$, and so $D_\calL \geq D_N$.
\end{proof}

%\begin{proof}
%Suppose $t>D_N$. Then there exist $\varepsilon>0$ and $x_\varepsilon \geq 1$ such that $t>D_N+\varepsilon>\log_x N_B(A,x)$, for all $x \geq x_\varepsilon$. So for some $C>0$, we have via Lemma \ref{lem:GeometricZetaFunctionIntegralTransformCountingFunction},
%\begin{align*}
%\zeta_\calL(t)
%	&\leq \int_{x_1}^\infty Cx^{D_N+\varepsilon}x^{-t-1}dx \\
%	&= \left[\frac{Cx^{D_N+\varepsilon-t}}{D_N+\varepsilon-t}\right]_{x_1}^\infty 
%	= 0 - \frac{Cx_1^{D_N+\varepsilon-t}}{D_N+\varepsilon-t},
%\end{align*}
%since $D_N+\varepsilon-t<0$. Hence, $\zeta_\calL(t)<\infty$.
%
%Now, suppose $t<D_N$. Then there exist $\varepsilon>0$ and  $x_\varepsilon \geq 1$ such that $D_N-\varepsilon<\log_x N_B(A,x)$, for all $x \geq x_\varepsilon$. So for $t<D_N-\varepsilon$, $r>x_1$, and some $c>0$, we have via Proposition \ref{prop:CountingFunctionD} and Lemma \ref{lem:GeometricZetaFunctionIntegralTransformCountingFunction},
%\begin{align*}
%\zeta_\calL(t) 
%	&\geq \int_{x_1}^r cx^{D_N-\varepsilon}x^{-t-1}dx 
%	= \left[\frac{cx^{D_N-\varepsilon-t}}{D_N-\varepsilon-t}\right]_{x_1}^r =:f(t).
%\end{align*}
%Now, since $D_N-\varepsilon-t>0$, we deduce that $f(t)$ is unbounded as $t \rightarrow \infty$, and hence that  $\zeta_\calL(t)$ diverges.
%\end{proof}

For a given fractal string $\calL$, Theorem \ref{thm:CountingFunctionAndComplexDimensions} (cf. Theorem 5.10 and Theorem 5.18 in \cite{LapvF6}) shows that under mild conditions the complex dimensions $\calD_\calL$ contain enough information to determine the geometric counting function $N_\calL$ (at least, asymptotically, and in important geometric situations, exactly\footnote{This is the case, for instance, for self-similar strings.}). 

%\begin{definition}
%\label{def:GeometricCountingFunction}
%Let $\calL$ be a fractal string and $x \in (0,\infty)$. Then the \emph{geometric %counting function} of $\calL$ is given by
%\begin{align}
%\label{eqn:GeometricCountingFunction}
%N_\calL(x)	&:= \#\{ j \geq 1 : \ell_j^{-1} \leq x\} = \sum_{n \geq 1 : l_n^{-1} \leq %x} m_n.
%\end{align}
%\end{definition}

\begin{theorem}
\label{thm:CountingFunctionAndComplexDimensions}
Let $\calL$ be a fractal string such that $\calD_\calL$ consists entirely of simple poles with respect to a window $W$. Then, under certain mild growth conditions on $\zeta_\calL$,\footnote{Namely, if $\zeta_\calL$ is \emph{languid} (see \cite[Def.~5.2]{LapvF6}) of a suitable order.} we have
\begin{align}
\label{eqn:CountingFunctionAndComplexDimensions}
N_\calL(x) = \sum_{\omega \in \calD_\calL}
\frac{x^{\omega}}{\omega}\textnormal{res}(\zeta_\calL(s);\omega) + \{\zeta_\calL(0)\} + R(x), 
\end{align}
where $R(x)$ is an error term of small order and the term in braces is included only if $0 \in W\backslash \calD_\calL$.
\end{theorem}

\begin{remark}
\label{rmk:SimpleAndStronglyLanguid}
If the poles are not simple, the explicit formula for $N_{\mathcal{L}}$ is slightly more complicated (see\cite[Chs.~5,6]{LapvF6}). If an ordinary fractal string $\Omega$ is {\it strongly languid} (see \cite[Def.~5.3]{LapvF6}), then by Theorem 5.14 and Theorem 5.22 of \cite{LapvF6}, Equation \eqref{eqn:CountingFunctionAndComplexDimensions} holds with no error term (i.e., $W=\C$ and $R(x) \equiv 0$) and hence, formula \eqref{eqn:CountingFunctionAndComplexDimensions} is exact in that case.
\end{remark}

\begin{remark}
\label{rmk:FractalTubeFormulas}
Similar (but harder to derive) explicit formulas called \emph{fractal tube formulas} are obtained in \cite[Ch.~8]{LapvF6}, which, as described therein, allows for the expression of $V_\textnormal{inner}(\varepsilon)$ in terms of the underlying (visible) complex dimensions of $\calL$. (Still in \cite{LapvF6}, they are used, in particular, to derive the equivalence of (i) and (iii) in Theorem \ref{thm:CriterionForMinkowskiMeasurability} below.) We will return to this topic in Section \ref{sec:SummaryOfResultsAndOpenProblems} when discussing Open Problems \ref{op:BoxCountingAndCountingFormulas} and \ref{op:VolumeTubularNeighborhood}.
\end{remark}

Analogous results regarding connections between the structure of the complex dimensions $\calD_\calL$ of an ordinary fractal string $\Omega$ with lengths $\calL$ and the (inner) Minkowski measurability of $\partial\Omega$ are presented in the next section.

\subsection{Classic results}
\label{sec:ClassicResults}

The following theorem is precisely Theorem 1.10 of \cite{LapvF6}. It is actually a consequence of a classic theorem of Besicovitch and Taylor (see \cite{BesTa}) stated in terms of ordinary fractal strings, and was first observed in this context in \cite{Lap2}.\footnote{There is, however, one significant difference with the setting of \cite{BesTa}. Namely, here, as in \cite{Lap2} and \cite{LapvF6}, we are assuming that we are working with the \emph{inner} (rather than ordinary) Minkowski dimension and Minkowski content of $\partial\Omega$; see the statement and the proof of Theorem 1.10 in \cite{LapvF6}, along with Equation \eqref{eqn:VolumeInnerNeighborhood} above. By contrast, in the context of \cite{BesTa}, one should assume that $\Omega$ is of full measure in its closed convex hull (i.e., in the smallest compact interval containing it).} 

\begin{theorem}
\label{thm:BesicovitchAndTaylor}
Suppose $\Omega$ is an ordinary fractal string with infinitely many lengths denoted by $\calL$. Then the abscissa of convergence of $\zeta_\calL$ coincides with the Minkowski dimension of $\partial\Omega$. That is, $D_\calL=\dim_M(\partial\Omega)$.
\end{theorem}

The following result is Theorem 8.15 of \cite{LapvF6}. For complete details regarding connections between complex dimensions and Minkowski measurability, see \cite[Ch.~8]{LapvF6}.

\begin{theorem}[Criterion for Minkowski measurability]
\label{thm:CriterionForMinkowskiMeasurability}
Let $\Omega$ be an ordinary fractal string whose geometric zeta function $\zeta_\calL$ has a meromorphic extension which satisfies certain mild growth conditions.\footnote{Specifically, $\zeta_\calL$ is languid for a screen $S$ passing strictly between the vertical line $\textnormal{Re}(s)=D_\calL$ and all the complex dimensions (of the corresponding fractal string) $\calL$ with real part strictly less than $D_\calL$, and not passing through 0.} Then the following are equivalent\emph{:}
\begin{enumerate}
	\item $D_\calL$ is the only complex dimension with real part $D_\calL$, and it is simple.
	\item $N_\calL(x)=cx^{D_\calL}+o(x^{D_\calL})$ as $x \to\infty$, for some positive constant $c$.\footnote{In the spirit of Proposition \ref{prop:CountingFunctionD}, condition (ii) is easily seen to be equivalent to 
\[
\ell_j=Lj^{-1/D_\calL}+o(j^{-1/D_\calL}) \quad \textnormal{as} \quad j\to\infty,
\]
for some positive constant $L$. In that case, we have $c=L^{D_\calL}$.}
	\item $\partial\Omega$ is Minkowski measurable.
\end{enumerate}
Moreover, if any of these conditions is satisfied, then the Minkowski content $\mink$
of $\partial\Omega$ is given by
\begin{align*}
%\label{eqn:MinkowskiContentFormula}
\mink	&= \frac{c2^{1-D_\calL}}{1-D_\calL} = 2^{1-D_\calL}\frac{\textnormal{res}(\zeta_\calL(s);D_\calL)}{D_\calL(1-D_\calL)}. 
\end{align*}
\end{theorem}

\begin{remark}
\label{rmk:TheoryOfComplexDimensions}
We note that the equivalence of (ii) and (iii) in Theorem \ref{thm:CriterionForMinkowskiMeasurability} was first established in \cite{LapPo1} for any ordinary fractal string, without any hypothesis on the growth of the associated geometric zeta function. As was alluded to in Remark \ref{rmk:FractalTubeFormulas}, however, the equivalence of (i) and (iii) in Theorem \ref{thm:CriterionForMinkowskiMeasurability} was proved in \cite{LapvF6} (and in earlier works of the authors of \cite{LapvF6}) by using a suitable generalization of Riemann's explicit formula that is central to the theory of complex dimensions and is obtained in \cite[Chs.~5 \& 8]{LapvF6}.
\end{remark}

\begin{example}[The Cantor set is not Minkowski measurable]
\label{eg:CantorSetNotMinkowskiMeasurable}
By Equation \eqref{eqn:CantorStringComplexDimensions} in Example \ref{eg:ComplexDimensionsOfCantorString}, there is an infinite collection of complex dimensions $\omega \in \calD_{CS}$ of the Cantor string with real part $D_{CS}=\log_{3}2$. Hence, by Theorem \ref{thm:CriterionForMinkowskiMeasurability}, the Cantor set $C$ is not Minkowski measurable. This fact was first established in \cite{LapPo1} by using the equivalence of (ii) and (iii) and showing that (ii) does not hold. Actually, still in \cite{LapPo1}, for $\alpha=\dim_{B}C=D_{CS}$, both $\mink^{\alpha*}=\mink^*$ and $\mink^{\alpha}_*=\mink_*$ are explicitly computed and shown to be different (with $0<\mink_*<\mink^*<\infty$). This result was significantly refined and extended in \cite[Ch.~10]{LapvF6} in the broader context of generalized Cantor strings.
\end{example}

\begin{remark}
\label{rmk:LatticeNonLattice}
Example \ref{eg:CantorSetNotMinkowskiMeasurable} is indicative of another result from \cite{LapvF6} pertaining to a dichotomy in the properties of self-similar attractors of certain iterated function systems on compact intervals. Specifically, if an iterated function system on a compact interval $I$ satisfies the open set condition with at least one gap and there is some $0<r<1$ and positive integers $k_j$ such that $\gcd(k_1,\ldots,k_N)=1$ and the scaling ratios satisfy $r_j=r^{k_j}$ for each $j = 1,\ldots, N$, then the complement $I \setminus A$ of the resulting attractor $A$ is an ordinary fractal string known as a \emph{lattice self-similar string}. For example, the Cantor string $\Omega_{CS}=[0,1] \setminus C$ is a lattice self-similar string. If no such $r$ exists, then $I \setminus A$ is a \emph{nonlattice self-similar string}. The complex dimensions of a self-similar string are given by (a subset of) the complex roots of the corresponding Moran equation \eqref{eqn:Moran}. In the lattice case there are countably many complex dimensions with real part $D_\calL=\dim_{B}A=\dim_{M}A$, so by Theorem \ref{thm:CriterionForMinkowskiMeasurability}, $A$ is not Minkowski measurable. In the nonlattice case, Theorem \ref{thm:CriterionForMinkowskiMeasurability} does not necessarily apply (because its hypotheses need not be satisfied, see \cite[Example~5.32]{LapvF6}), however the only complex dimension with real part $D_\calL$ is $\dim_{B}A=\dim_{M}A$ and by Theorem 8.36 of \cite{LapvF6} we have that $A$ is Minkowski measurable. Therefore, the boundary of a self-similar string is Minkowski measurable if and only if it is nonlattice. See \cite[\S 8.4]{LapvF6} for details.
\end{remark}
 
We conclude the section on classic results with the following remark which, in light of the expression for $V_{\textnormal{inner}}(\varepsilon)$ obtained in \cite{LapPo1} (see also \cite[Eq.~(8.1)]{LapvF6}), can be deduced from Lemma 1 of \cite[\S 1.4]{Lev}.\footnote{For convenience, Remark \ref{rmk:DimensionEqualsOrder}
% Definition \ref{def:OrderOfCountingFunction} and Lemma \ref{lem:DimensionEqualsOrder} are 
is stated in the language of fractal strings. A direct (and independent) proof of the fact that Equation \eqref{eqn:DimensionEqualsOrder} holds 
%Lemma \ref{lem:DimensionEqualsOrder} 
can be found in \cite{LapvF6}. (See also \cite{LapLuvF1}.) Also, Equation \eqref{eqn:DimensionEqualsOrder} is a simple consequence of results obtained in \cite{LapPo1}.} The remark below provides yet another connection between counting functions and dimensions, although it is a simple restatement of Proposition \ref{prop:CountingFunctionDFractalStringD} in our context.

%%%%%%%%%%%%%%
\begin{remark}
\label{rmk:DimensionEqualsOrder}
%\begin{definition}
%\label{def:OrderOfCountingFunction}
For a fractal string $\calL$, the \emph{order of the geometric counting function} $N_\calL$, denoted $\rho_\calL$, is given by 
\begin{align}
\label{eqn:OrderOfCountingFunction}
\rho_\calL	&:=\limsup_{x \rightarrow \infty} \frac{\log{N_\calL(x)}}{\log{x}}.
\end{align}  
%\end{definition}
We have that the dimension $D_\calL$ coincides with $\rho_\calL$. That is,
%\begin{lemma}
%\label{lem:DimensionEqualsOrder}
%For a fractal string $\calL$,
\begin{align}
\label{eqn:DimensionEqualsOrder}
D_\calL	&=\rho_\calL.
\end{align}
%\end{lemma}
Note that, for a given fractal string $\calL$, the order of the counting function $\rho_\calL$ given in Equation \eqref{eqn:OrderOfCountingFunction} and the value $D_N$ given in Equation \eqref{eqn:CountingFunctionD} provide essentially the same information regarding the geometric counting function $N_\calL$. Indeed, it can be shown directly that $\rho_\calL=D_N$, and hence Equation \eqref{eqn:DimensionEqualsOrder} would follow from Proposition \ref{prop:CountingFunctionDFractalStringD}. This connection is examined further in \cite{LapRoZu}.
\end{remark}
%%%%%%%%%%%%%

In the next section, motivated by the box-counting function $N_B$ and connections between the geometric counting function $N_\calL$ and dimension $D_\calL$ of a fractal string $\calL$, we define and investigate the properties of box-counting fractal strings.

\section{Box-counting fractal strings and zeta functions}
\label{sec:BoxCountingFractalStringsAndZetaFunctions}

In this section, we develop the definition of and results pertaining to \emph{box-counting fractal strings}. These fractal strings are defined in order to provide a framework in which one may, perhaps, extend the results on ordinary fractal strings via associated zeta functions and complex dimensions in \cite{LapvF6} to bounded sets. Further exploration with box-counting fractal strings, such as Minkowski measurability of bounded sets, is central to the development of the authors' paper \cite{LapRoZu}. The box-counting fractal string and the box-counting zeta function for bounded sets in Euclidean spaces were introduced by the second author during the First International Meeting of the Permanent International Session of Research Seminars (PISRS) at the University of Messina, PISRS Conference 2011: Analysis, Fractal Geometry, Dynamical Systems, and Economics.\footnote{See Remark \ref{rmk:PrecursorToBoxCountingZetaFunction}.} The introduction took place after listening to a lecture of the third author about his results (with the first author and Goran Radunovi\'c) in \cite{LapRaZu} on distance and tube zeta functions for arbitrary compact subsets of $\R^m$. Some of these results are also discussed in Section \ref{sec:DistanceAndTubeZetaFunctions} below.

\begin{remark}
\label{rmk:PrecursorToBoxCountingZetaFunction}
At the time, the second author did not yet know that the first author had already proposed (since the early 2000s, in several research documents) to introduce and study a `box-counting zeta function' (defined via the Mellin transform of a `box-counting function'), much as in Corollary \ref{cor:EquivalentZetaFunctions} and Definition \ref{def:BoxCountingFunction} as well as Remark \ref{rmk:VariousCountingFunctions}) in order to develop a higher-dimensional theory of complex dimensions and obtain the associated fractal tube formulas for suitable fractal subsets of $\R^m$ (established when $m=1$ in \cite[Ch.~8]{LapvF6}). The explicit construction of a `box-counting fractal string' associated with a given bounded set $A\subset\R^m$ is new and potentially quite useful, however.\footnote{Compare, when $m=1$, the construction provided in \cite{LapMa} of certain fractal strings associated with suitable monotonically increasing step functions, viewed as `geometric counting functions' (in the sense of \cite{LapPo1,LapvF6} and of Definition \ref{def:GeometricCountingFunction} above).} Of course, the other results of \cite{LapRoZu} described in Section \ref{sec:BoxCountingFractalStringsAndZetaFunctions} are new as well.
\end{remark}

\subsection{Definition of box-counting fractal strings}
\label{sec:DefinitionOfBoxCountingFractalStrings}

If $A \subset \R^m$ is bounded, then the diameter of $A$, denoted $\diam(A)$, is finite. So for nonempty $A$ and all $x$ small enough, we have $N_B(A,x)=1$ when $N_B(A,\cdot)$ is given as in Definition \ref{def:BoxCountingFunction} or one of the options in Remark \ref{rmk:VariousCountingFunctions}. Indeed, for a given bounded infinite set $A$, each such box-counting function uniquely defines a fractal string $\calL_B$, which is introduced below and called the \emph{box-counting fractal string}, by uniquely determining a sequence of distinct scales $(l_n)_{n\in\N}$ along with corresponding multiplicities $(m_n)_{n\in\N}$. 

Given a fixed bounded infinite set $A$, the range of a chosen box-counting function $N_B(A,\cdot)$ can be thought of as a strictly increasing sequence of positive integers $(M_n)_{n\in\N}$. In this context, we can readily define a fractal string $\calL_B$ whose geometric counting function $N_{\calL_B}$ essentially coincides with $N_B(A,\cdot)$; see Lemma \ref{lem:CountingFunctionCorrelation} below. To this end, the key idea is to make the distinct scales $l_n$ of the desired (box-counting) fractal string $\calL_B$ correspond to the scales at which the box-counting function $N_B(A,\cdot)$ jumps. Furthermore, the multiplicities $m_n$ are defined in order to have the resulting geometric counting function $N_{\calL_B}$ (nearly) coincide with the chosen box-counting function $N_B(A,\cdot)$. Such \emph{box-counting fractal strings} potentially allow for the development of a theory of complex dimensions of fractal strings, as presented in \cite{LapvF6}, by means of results in Section \ref{sec:FractalStringsAndZetaFunctions} similar to Theorem \ref{thm:CountingFunctionAndComplexDimensions} above.\footnote{Parts of such a theory are now developed in \cite{LapRaZu}, using distance and tube zeta functions; see Section \ref{sec:DistanceAndTubeZetaFunctions} below for a few sample results. Once the two parallel theories are more fully developed, a challenging problem will consist in comparing and contrasting their respective results and scopes.} These concepts are central to the development of the paper \cite{LapRoZu}.

\begin{definition} 
\label{def:BoxCountingFractalString} 
Let $A$ be a bounded infinite subset of $\R^m$ and let $N_B(A,\cdot)$ denote a box-counting function given by one of the options in Remark \ref{rmk:VariousCountingFunctions}. Denote the range of $N_B(A,\cdot)$ as a strictly increasing sequence of positive integers $(M_n)_{n\in\N}$. 
For each $n\in\N$, let $l_n$ be the scale given by
\begin{align}
\label{eqn:LengthsOfABoxCountingFractalString}
l_n:=(\sup\{x \in (0,\infty): N_B(A,x)=M_n\})^{-1}.
\end{align}
Also, let $m_1:=M_2$, and for $n\geq 2$, let $m_n:=M_{n+1}-M_n$. 
The \emph{box-counting fractal string} of A, denoted $\calL_B$, is given by
\begin{align*}
%\label{eqn:BoxCountingFractalString}
\calL_B &:= \{l_n : l_n \textnormal{ has multiplicity } m_n, n \in \N\}.
\end{align*} 
\end{definition}

\begin{remark}
\label{rmk:CountingFunctionDefinesUniqueFractalString}
Note that the distinct scales $l_n$ and the multiplicities $m_n$ are \emph{uniquely defined by the box-counting function} $N_B(A,\cdot)$ since $N_B(A,x)$ is nondecreasing as $x\rightarrow\infty$. Also, each $l_n$ is equivalently given by
\[ 
l_n = \inf\{\varepsilon \in (0,\infty): N_B(A,\varepsilon^{-1})=M_n\}.
\]
\end{remark}

It remains to show that $\calL_B$ is indeed a fractal string; see Definition \ref{def:FractalString}. That is, since we want to use as many of the results from \cite{LapvF6} as possible (some of which are presented in Section \ref{sec:FractalStringsAndZetaFunctions}), we must verify that $\calL_B=(\ell_j)_{j\in\N}$ is a nonincreasing sequence of positive real numbers which tends to zero. This is accomplished with the following proposition, in which other behaviors of $N_B(A,\cdot)$ are also determined. 

For clarity of exposition and in order to ease the notation used in this section, in particular in the following proposition, take $N_B(A,\cdot)$ to be defined by option (i) of Remark \ref{rmk:VariousCountingFunctions} and let $N_B(A,0):=0$. (Completely analogous results hold when $N_B(A,\cdot)$ is given by Definition \ref{def:BoxCountingFunction} or one of the other options in Remark \ref{rmk:VariousCountingFunctions}, \emph{mutatis mutandis}.) Note that we have $N_B(A,x)\leq N_B(A,y)$ whenever $0<x<y$. In this setting, and following Remark \ref{rmk:varepsilonVSx}, $x^{-1}$ denotes the diameter of the sets used to cover $A$. Furthermore, let $x_n:=l_n^{-1}$ for each $n\in\N$, and note that we have $N_B(A,x_2)=m_1=M_2$ and 
\[
N_B(A,x_{n+1})-N_B(A,x_n)=m_n=M_{n+1}-M_n, \quad \textnormal{for } n \geq 2.
\]

\begin{proposition}
\label{prop:BoxCountingJumps}
Let $A$ be a bounded infinite subset of $\R^m$ and let $l_n$ be given by Equation \emph{\eqref{eqn:LengthsOfABoxCountingFractalString}}. Then the sequence $(x_n)_{n\in\N}:=(l_n^{-1})_{n\in\N}$ is a countably infinite, strictly increasing sequence of positive real numbers such that, for each $n \in \N$ and all $x$ such that $x_{n-1} < x \leq x_n$ \emph{(}letting $x_0=0$\emph{)}, we have
\begin{align}
\label{eqn:BoxCountingJumps}
N_B(A,x_{n-1}) &< N_B(A,x) = N_B(A,x_n).
\end{align}
Furthermore, 
\begin{enumerate}
\setlength{\itemsep}{0in}
\item $x_1>0$ and $N_B(A,x_1)=1$,
\item $x_n \nearrow \infty$ as $n \rightarrow \infty$, and
\item $\displaystyle \bigcup_{n \in \N} \{N_B(A,x_n)\} = \textnormal{range}~N_B(A,\cdot)$.  
\end{enumerate}
\end{proposition}

\begin{proof}
We have that $N_B(A,x)$ is nondecreasing for $x>0$. Further, the range of $N_B(A,\cdot)$, denoted $\textnormal{range}~N_B(A,\cdot)$ (and also realized as the sequence $(M_n)_{n\in\N}$ above), is at most countable since it is a subset of $\N$. In fact, $\textnormal{range}~N_B(A,\cdot)$ is countably infinite (otherwise, $A$ would be finite). Hence, $(x_n)_{n \in \N}$ is a unique, countably infinite, strictly increasing sequence of positive real numbers such that, for each $n \in \N$ and all $x$ such that  $x_{n-1} < x \leq x_n$ (letting $x_0=0$), we have
\begin{align*}
%\label{eqn:BoxCountingJumpsAgain}
N_B(A,x_{n-1}) &< N_B(A,x) = N_B(A,x_n).
\end{align*}

Since $A$ is bounded and contains more than two elements, there exists a unique $x'\in (0,\infty)$ such that $N_B(A,x)=1$ if $0<x \leq x'$, and $N_B(A,x)>1$ if $x>x'$. By the definition of the sequence $(x_n)_{n\in\N}$, we have $x'=x_1$.

Now, suppose $(x_n)_{n\in\N}$ has an accumulation point at some $x'' \in (0,\infty)$. Then $N_B(A,x'')=\infty$ since $N_B(A,\cdot)$ increases by some positive integer value at $x_n$ for each $n \in \N$ and since $\textnormal{range}~N_B(A,\cdot)\subset\N$. However, this contradicts the boundedness of $A$. Further, assuming $N_B(A,\cdot)$ is bounded implies that $A$ is finite. Hence, $x_n \nearrow \infty$ as $n \rightarrow \infty$.

Lastly, suppose there exists $k \in \textnormal{range}~N_B(A,\cdot)$ such that we have $k \neq N_B(A,x_n)$ for all $n \in \N$. Since $x_n \nearrow \infty$ as $n \rightarrow \infty$ and $N_B(A,\cdot)$ is nondecreasing, there exists a unique $n_0 \in \N$ such that $x_{n_0-1}<y<x_{n_0}$ for all $y$ such that $N_B(A,y)=k$. However, Equation \eqref{eqn:BoxCountingJumps} implies $N_B(A,y)=k=N_B(A,x_{n_0})$, which is a contradiction. Therefore, $\bigcup_{n \in \N}\{N_B(A,x_n)\} = \textnormal{range}~N_B(A,\cdot)$.
\end{proof}

\begin{remark}
\label{rmk:BoxCountingFractalStringIsAFractalString}
By line (ii) of Proposition \ref{prop:BoxCountingJumps}, $l_n \searrow 0$ as $n \rightarrow \infty$ and, hence, $\calL_B$ is indeed a fractal string in the sense of Definition \ref{def:FractalString}.
\end{remark}

%\begin{remark}
%\label{rmk:PropositionHoldsWithFiniteA}
%If $A$ is a finite set with at least two elements, then the results of Proposition \ref{prop:BoxCountingJumps} hold except for (2) and the existence of a countably infinite sequence of such $x_n \in (0,\infty)$. More specifically, if $A$ comprises $K \geq 2$ elements, then there exists $(x_n)_{n=1}^{N_0}$ such that for each $n=1,\ldots,N_0$ and all $x$ where $x_{n-1} < x \leq x_n$ (letting $x_0=0$), we have
%\begin{align*}
%\label{eqn:BoxCountingJumpsFinite}
%N_B(A,x_{n-1}) &< N_B(A,x) = N_B(A,x_n).
%\end{align*} 
%Moreover, for $x>x_{N_0}, N_B(A,x)=N_B(A,x_{N_0})=K$. 
%\end{remark}

\begin{example}[Box-counting fractal string of the Cantor set]
\label{eg:CantorSetBoxCountingFractalString}
Consider the Cantor set $C$. For $x>0$, let the box-counting function $N_B(C,x)$ be the minimum number of sets of diameter $x^{-1}$ required to cover $C$ (i.e., as in option (i) of Remark \ref{rmk:VariousCountingFunctions}). Then the box-counting fractal string $\calL_B$ of $C$ is given by
\begin{align}
\label{eqn:CantorSetBoxCountingFractalString}
\calL_B	&= \{l_1=1: m_1=2\}\cup\{l_n=3^{-(n-1)}: m_n =2^{n-1}, n \geq 2 \}.
\end{align}
Indeed, for each $n \in \N$, exactly $2^{n}$ intervals of diameter $3^{-n}$ are required to cover $C$. If $x^{-1}<3^{n}$, then more than $2^{n}$ intervals of diameter $x^{-1}$ are required to cover $C$. 
\end{example}

\begin{example}[Box-counting fractal string of a 1-dimensional fractal]
\label{eg:BoxCountingFractalStringOf1DFractal}
Consider the self-similar set $F$ which is the attractor of the IFS $\bfPhi_1=\{\Phi_j\}_{j=1}^4$ on the unit square $[0,1]^2\subset\R^2$ given by
\begin{align*}
\Phi_1(x)	&=\frac{1}{4}x, \quad \Phi_2(x)	=\frac{1}{4}x+\left(\frac{3}{4},0\right), \quad
\Phi_3(x)	=\frac{1}{4}x+\left(\frac{3}{4},\frac{3}{4}\right), \quad \textnormal{and}\\ \Phi_4(x)	&=\frac{1}{4}x+\left(0,\frac{3}{4}\right).
\end{align*}
The Moran equation of $F$ is simply $4\cdot 4^{-s}=1$, hence $D_{\bfPhi_1}=\dim_BF=\dim_MF=1$ and $F$ is a 1-dimensional self-similar set which is totally disconnected.

Let $N_B(F,x)$ be as in Definition \ref{def:BoxCountingFunction}. Then, for $x \in (0,2]$, we have
\[
N_B(F,x)=
\begin{cases}
1, & 0 < x \leq 2/\sqrt{2}, \\
2, & 2/\sqrt{2} < x \leq 8/\sqrt{17}, \\
3, & 8/\sqrt{17} < x \leq 2.
\end{cases}
\]
Indeed, we have: $\sqrt{2}$ is the distance between $(0,0)$ and $(1,1)$; $\sqrt{17}/4 $ is the minimum distance between $(0,0),(1,1/4),$ and $(1/4,1)$; and the minimum distance between $(0,0),(0,1),(1,0),$ and $(1,1)$ is $1$. Hence, $M_1=1, M_2=2$, and $M_3=3$. 

For $x>2$, the self-similarity of $F$ implies that $M_n=j4^k$, where $n$ is uniquely expressed as $n=3k+j$ with $k\in \N\cup\{0\}$ and $j\in\{1,2,3\}$. So, for $n\geq 2$, we have $m_n=M_{n+1}-M_n=4^k$ and therefore the box-counting fractal string $\calL_B=(\ell_j)_{j=1}^\infty$ of $F$ is the sequence obtained by putting the following collection of distinct scales in nonincreasing order and listing them according to multiplicity:
\begin{align*}
%\label{eqn:1DFractalBoxCountingFractalString}
&\left\{\sqrt{2}/2: \textnormal{ multiplicity } 2 \right\}\\ 
&\cup \left\{\sqrt{2}/(2\cdot 4^k): \textnormal{ multiplicity } 4^k, k\in \N \right\}\\
&\cup \left\{\sqrt{17}/(8\cdot 4^k): \textnormal{ multiplicity } 4^k, k\in \N\cup\{0\} \right\}\\
&\cup \left\{1/(2\cdot 4^k): \textnormal{ multiplicity } 4^k, k\in \N\cup\{0\} \right\}. 
\end{align*}
\end{example}

Examples \ref{eg:CantorSetBoxCountingFractalString} and \ref{eg:BoxCountingFractalStringOf1DFractal} will be revisited and expanded upon in the following subsection.

%%%
%%%
%%%
%%%
%%%
%%%

\subsection{Box-counting zeta functions}
\label{sec:BoxCountingZetaFunctions}

Suppose $A$ is a bounded infinite subset of $\R^m$. Each scale $l_n \in \calL_B$ is distinct and, for $n \geq 2$, counted according to the multiplicity $m_n:=N_B(A,x_{n+1})-N_B(A,x_n)$. It will help to note that we can also consider $\calL_B$ to be given by the nonincreasing sequence $(\ell_j)_{j \in \N}$, where the distinct values among the $\ell_j$'s repeat the $l_n$'s according to the multiplicities $m_n$. (The convention of distinguishing the notation $\ell_j$ and $l_n$ in this way is established in \cite{LapvF6} and its predecessors, where the distinction allows for various results therein to be more clearly stated and derived.) In this setting, we immediately have the following connection between $N_{\calL_B}$, the counting function of the reciprocal lengths of $\calL_B$, and the box-counting function $N_B(A,x)$.

\begin{lemma}
\label{lem:CountingFunctionCorrelation}
For $x \in (x_1,\infty)\setminus(x_n)_{n \in \N}$,
\begin{align*}
N_{\calL_B}(x)=N_B(A,x).
\end{align*}
\end{lemma}

\begin{proof}
The result follows at once from Definitions \ref{def:BoxCountingFunction} and \ref{def:BoxCountingFractalString}.
\end{proof}

In general, Lemma \ref{lem:CountingFunctionCorrelation} does not hold for $x=x_n$, though equality may hold for certain choices of $N_B(A,x)$. Furthermore, the primary applications of  Lemma \ref{lem:CountingFunctionCorrelation} are Corollary \ref{cor:EquivalentZetaFunctions} and Theorem \ref{thm:JohnnyDarko} where the behavior of $N_B(A,x)$ at $x=x_n$ does not affect the conclusions. Moreover, for a bounded infinite set $A$, the geometric zeta function of the box-counting fractal string $\calL_B$ is given by
\begin{align*}
\zeta_{\calL_B}(s)&=N_B(A,l_2^{-1})l_1^s+\sum_{n=2}^{\infty} (N_B(A,l_{n+1}^{-1})-N_B(A,l_n^{-1})) l_n^s = \sum_{j=1}^\infty \ell_j^s,
\end{align*}
for $\textnormal{Re}(s)>D_{\calL_B}$. We take this zeta function to be our \emph{box-counting zeta function} for a bounded infinite set $A$ in Definition \ref{def:BoxCountingZetaFunction}.

\begin{definition}
\label{def:BoxCountingZetaFunction} 
Let $A$ be a bounded infinite subset of $\R^m$. The \emph{box-counting zeta function} of $A$, denoted $\zeta_B$, is the geometric zeta function of the box-counting fractal string $\calL_B$. That is, 
\begin{align*}
%\label{eqn:BoxCountingZetaFunction}
\zeta_B(s) := \zeta_{\calL_B}(s) = \sum_{n=1}^\infty m_n l_n^s, 
\end{align*}
for $Re(s) > D_B := D_{\calL_B}$. The (optimum) value $D_B$ is the \emph{abscissa of convergence} of $\zeta_B$. The set of \emph{box-counting complex dimensions} of $A$, denoted $\calD_B$, is the set of complex dimensions $\calD_{\calL_B}$ of the box-counting fractal string $\calL_B$.
% If $A$ is finite with $K \geq 2$ points, the \emph{box-counting zeta function} is given by
%\begin{align*}
%\label{eqn:BoxCountingZetaFunctionFinite}
%\zeta_B(s) := \sum_{n=1}^{K-1} m_n l_n^s,
%\end{align*}
%for $s \in \C$. If $A$ is a singleton or the empty set, the \emph{box-counting zeta function} is identically zero (i.e., $\zeta_B(s) \equiv 0, s \in \C$). If $A$ is finite, set $D_B=0$ and $\calD_B=\emptyset$.
\end{definition}

\begin{remark}
\label{rmk:FiniteBoxCountingFractalString}
Note that we do not consider the case when $A$ is finite. One may, of course, define the box-counting fractal string $\calL_B$ for such a set as a finite sequence of positive real numbers.  In that case, however, the box-counting zeta function would comprise a finite sum, which would yield an abscissa of convergence $-\infty$ and no complex dimensions; see Remark \ref{rmk:AbscissaForFiniteSequences}. That is, in the context of the theory of complex dimensions of fractal strings, the case of finite sets is not very interesting.
%then $\calL_B$ comprises a finite number of scales. In that case, $\calL_B$ is not a fractal string in the sense of Definition \ref{def:FractalString}. Nonetheless, we allow for finite sets in the setting of box-counting fractal strings and, consequently, we must distinguish between finite and infinite sets in the definition of the box-counting zeta function given in Definition \ref{def:BoxCountingZetaFunction}. 
\end{remark}

\begin{example}[Box-counting zeta function of the Cantor set]
\label{eg:CantorSetBoxCountingZetaFunction}
By Example \ref{eg:CantorSetBoxCountingFractalString}, the box-counting fractal string $\calL_B$ of the Cantor set $C$ is given by Equation \eqref{eqn:CantorSetBoxCountingFractalString}.
%\begin{align}
%\label{eqn:CantorSetBoxCountingFractalString}
%\calL_B	&= \{3^{n-1}: m_1 =2, m_n =2^{n-1}, n \geq 2 \}.
%\end{align}
It follows that for $\textnormal{Re}(s)>\log_{3}2$, the box-counting zeta function of $C$ is given by
\begin{align*}
%\label{eqn:CantorSetBoxCountingZetaFunction}
\zeta_B(s)	&= 2+\sum_{n=2}^\infty 2^{n-1}\cdot 3^{-(n-1)s} 
%= 2+\frac{2\cdot3^{-s}}{1-2\cdot3^{-s}} 
= 1+\frac{1}{1-2\cdot3^{-s}}. 
\end{align*}
Thus, $D_B=\dim_{B}C=\dim_{M}C=\log_{3}2$ and $\zeta_B$ has a meromorphic extension to all of $\C$ given by the last expression in the above equation. Moreover, we have 
\begin{align*}
%\label{eqn:CantorSetBoxCountingComplexDimensions}
\calD_B	&=\calD_{CS}=\calD_{\calL_{CS}}=\left\{ \log_3{2} + i\frac{2\pi}{\log{3}}z : z \in \Z \right\}.
\end{align*}
\end{example}

\begin{example}
\label{eg:BoxCountingZetaFunction1DFractal}
The box-counting fractal string $\calL_B$ of the 1-dimensional self-similar set $F$ (the attractor of the IFS $\bfPhi_1$), where $N_B(F,x)$ is as in Definition \ref{def:BoxCountingFunction}, is given in Example  \ref{eg:BoxCountingFractalStringOf1DFractal}. Hence, the box-counting zeta function of $F$ is given (for $\textnormal{Re}(s)>1$) by
\begin{align}
\label{eqn:1DBoxCountingZetaFunction}
\zeta_B(s)	&= \left(\frac{\sqrt{2}}{2}\right)^s + \left(\left(\frac{\sqrt{2}}{2}\right)^s + \left(\frac{\sqrt{17}}{8}\right)^s + \left(\frac{1}{2}\right)^s\right)
\sum_{k=0}^\infty\left(\frac{4}{4^s}\right)^k \\
&= \left(\frac{\sqrt{2}}{2}\right)^s + \frac{\left(\sqrt{2}/2\right)^s + \left(\sqrt{17}/8\right)^s + \left(1/2\right)^s}{1-4\cdot 4^{-s}}.
\end{align}
Thus, $D_B=\dim_{B}F=\dim_{M}F=1$ and $\zeta_B$ has a meromorphic extension to all of $\C$ given by the last expression in the above equation. Moreover,  we have 
\begin{align}
\label{eqn:1DBoxCountingComplexDimensions}
\calD_B	&= \calD_B(\C)=\left\{1 + i\frac{2\pi}{\log{4}}z : z \in \Z \right\}.
\end{align}
Note that the series corresponding to $\zeta_B(1)$ is divergent. Hence, the fractal string $\calL_B$ does not correspond to an ordinary fractal string (which, by definition, requires $\zeta_\calL(1)=\sum_{j=1}^\infty \ell_j$ to be convergent). 
\end{example}

\begin{remark}
\label{rmk:LatticeSelfSimilarIFS}
The Cantor set $C$ and the 1-dimensional self-similar set $F$ are each the attractor of a \emph{lattice} iterated function system; see \cite[\S 13.1]{LapvF6} as well as \cite{LapPe3,LapPe2,LapPeWin}. Essentially, an IFS generated by similarities (i.e., an IFS for which Equation \eqref{eqn:IFSSelfSimilar} holds) is \emph{lattice} if, for the distinct values $t_1,\ldots,t_{N_0}$ among the scaling ratios $r_1,\ldots,r_N$, there are positive integers $k_j$ where $\gcd(k_1,\ldots,k_{N_0})=1$ and a positive real number $0<r<1$ such that $t_j=r^{k_j}$ for each $j=1,\ldots,N_0$.\footnote{Equivalently, the precise definition (following \cite{LapvF6}) is that the \emph{distinct} scaling ratios generate a multiplicative group of rank 1.} Note that in each case, the box-counting complex dimensions comprise a set of complex numbers with a unique real part (equal to the box-counting dimension) and a vertical (and arithmetic) progression, in both directions, of imaginary parts. 

In the case of the Cantor set $C$, the box-counting complex dimensions $\calD_B$ coincide with the usual complex dimensions $\calD_{CS}$. Moreover, the structure of $\calD_{CS}$ allows for the application of Theorem \ref{thm:CriterionForMinkowskiMeasurability} and, hence, we conclude (as in \cite{LapPo1} and \cite{LapvF6}) that $C$ is not Minkowski measurable.

In the case of the 1-dimensional self-similar set $F$ of Examples \ref{eg:BoxCountingFractalStringOf1DFractal} and \ref{eg:BoxCountingZetaFunction1DFractal}, the set of complex dimensions $\calD_B$ has no counterpart in the context of usual complex dimensions since $F$ is not the complement of an ordinary fractal string. As such, Theorem \ref{thm:CriterionForMinkowskiMeasurability} does not apply. Moreover, since $\dim_MF=1$, the corresponding results in \cite[\S 13.1]{LapvF6} do not apply either. 
%(Fractals with nonnegative integer Minkowski dimension are not considered therein.) 
This provides motivation for developing a theory of complex dimensions which can take such examples, and many others, into account. The box-counting fractal strings defined in this paper, and investigated further in \cite{LapRoZu}, provide a first step in developing one such theory. Analogous comments regarding the further development of a higher-dimensional theory of complex dimensions can be made about the results of \cite{LapRaZu} to be discussed in Section \ref{sec:DistanceAndTubeZetaFunctions}.
\end{remark}

The next corollary follows readily from Lemma \ref{lem:CountingFunctionCorrelation} and Proposition \ref{prop:CountingFunctionDFractalStringD}. It establishes the equivalence of the box-counting zeta function $\zeta_B$ and an integral transform of the (appropriately truncated) box-counting function $N_B(A,x)$.

\begin{corollary}
\label{cor:EquivalentZetaFunctions}
Let $A$ be a bounded set. Then 
\begin{align*}
%\label{eqn:EquivalentZetaFunctionsCorollary}
\zeta_B(s) &= \zeta_{\calL_B}(s) = s\int_{l_1^{-1}}^{\infty}x^{-s-1}N_B(A,x)dx,
\end{align*}
for $Re(s) > D_B$.
\end{corollary}

We close this subsection with a theorem which is a partial statement of our main result, Theorem \ref{thm:MainResult}. Specifically, the upper box-counting dimension of a bounded infinite set is equal to the abscissa of convergence of the corresponding box-counting zeta function. 

\begin{theorem}
\label{thm:JohnnyDarko}
Let $A$ be a bounded infinite subset of $\R^m$. Then $\overline{\dim}_{B}A=D_B$.
\end{theorem}

\begin{proof}
The proof follows from a connection made through $D_N$, the asymptotic growth rate of the geometric counting function $N_\calL(x)$ of a fractal string $\calL$ given by Equation \eqref{eqn:CountingFunctionD}. 

Let $\calL=\calL_B$. By Proposition \ref{prop:CountingFunctionDFractalStringD}, we have $D_\calL=D_B=D_N$. Now, Lemma \ref{lem:CountingFunctionCorrelation} implies $N_\calL(x)=N_B(A,x)$ for $x \in (l_1^{-1},\infty)\setminus(l_n^{-1})_{n\in\N}$. Since these counting functions are nondecreasing, the equation $\overline{\dim}_{B}A=D_N$ follows from the formulation  of $\overline{\dim}_{B}A$ given in Remark \ref{rmk:EquivalentUpperBoxDimension}.
\end{proof}

\subsection{Tessellation fractal strings and zeta functions} 
In this subsection, we loosely discuss another type of fractal string defined for a given bounded infinite subset $A$ of $\R^m$. Unlike the box-counting fractal string $\calL_B$, which is completely determined by the set $A$ and the box-counting function $N_B(A,\cdot)$, the \emph{tessellation fractal string} defined here depends on the set $A$, a chosen parameter, and a chosen family of tessellations of $\R^m$.
%Let $A$ be a bounded, infinite subset of $\R^m$. The associated box-counting string $\cal L_B=(\ell_j)_{j\in\N}$ is defined as follows.

First, choose a scaling parameter $\lambda\in(0,1)$. For any $n\in\N$, consider the $n$-th tessellation of $\R^m$ defined by the family of cubes of length $\lambda^n$ (obtained by taking translates of the cube $[0,\lambda^n]^m$ in $\R^m$). Henceforth, the number of cubes of the $n$-th tessellation that intersect $A$ is denoted by $m_n(\lambda)$. Let the scale $l_n(\lambda):=\lambda^n$ be of multiplicity $m_n(\lambda)$. This defines the box-counting fractal string $\calL(A,\lambda)=(\ell_j)_{j\in\N}$, where $(\ell_j)_{j\in\N}$ is the sequence starting with $l_1$ with multiplicity $m_1$, $l_2$ with multiplicity $m_2$, and so on. The geometric counting function $N_{\calL(A,\lambda)}(x)=\#\{j\in\N: \ell_j^{-1}\ge x\}$ of the fractal string is then well defined. 

More generally, let $U$ be a compact subset of $\R^m$ with nonempty interior, satisfying the following properties:
\begin{enumerate}
\addtolength{\itemsep}{0.3\baselineskip}
\item[(a)] there exists a countable family of isometric maps $f_j:\R^m\to\R^m$ such that the the family of sets $V_j=f_j(U)$, $j\ge1$, is a cover of $\R^m$;
\item[(b)] for $j\neq k$, the interiors of $V_j$ and $V_k$ are disjoint.
\end{enumerate}
We say that the family $(V_j)_{j\ge1}$ is a \emph{tessellation of $\R^m$}, generated by the basic shape (or `tile') $U$ and the family of isometries. Also, we say in short that \emph{$U$ tessellates~$\R^m$}. Note that if $\lambda\in(0,1)$ is a fixed real number, then the basic shape generates a sequence of tessellations indexed by $n\in\N$: $(\lambda^n V_j)_{j\ge1}$. The family $(\lambda^n V_j)_{j\ge1}$ is called the \emph{$n$-th tessellation of $\R^m$}, generated by the basic shape (or tile) $U$, the family of isometries $(f_j)_{j\ge1}$, and $\lambda\in(0,1)$.

Let $A$ be a given bounded set in $\R^m$. Define $m_n(\lambda)=m_n(A,U,\lambda)$ analogously as above, by counting the number of elements of the $n$-th tessellation which intersect $A$. The \emph{tessellation fractal string} $\calL(A,U,\lambda)$ of the set $A$ is then  the fractal string defined by
\begin{align*}
\calL(A,U,\lambda)	&:=\{l_n(\lambda)=\lambda^n:\mbox{ $l_n(\lambda)$ has multiplicity $m_n(\lambda)$, $n\in\N$}\}=(\ell_j)_{j\in\N}.
\end{align*}
The middle set is in fact a \emph{multiset}, by which we mean that its elements repeat with prescribed multiplicity. Using arguments analogous to those from \cite[pp.\ 38--39]{Falc}, we obtain that
\begin{equation}
\label{dimBA}
\overline\dim_BA=\limsup_{n\to\infty}\frac{\log m_n(\lambda)}{\log\lambda^{-n}}.
\end{equation}
(Also, see \cite[p.\ 41]{Falc} or \cite[p.\ 24]{Tri}.) Here,  we have also used a version of Proposition~\ref{prop:GeometricSequenceSufficient} above.

The geometric zeta function of the tessellation fractal string $\calL(A,U,\lambda)$, called the \emph{tessellation zeta function}, is given by 
\begin{align}
\label{zeta_tess}
\zeta_{\calL(A,U,\lambda)}(s)	&=\sum_{j=1}^{\infty}\ell_j^s=\sum_{n=1}^{\infty} m_n(\lambda)\,\lambda^{ns}
\end{align}
for $\textnormal{Re}(s)$ large enough. Also, when defined accordingly, the set of complex dimensions $\calD_{\calL(A,U,\lambda)}$ of the tessellation fractal string $\calL(A,U,\lambda)$ is called the set of \emph{tessellation complex dimensions} of $A$ (relative to the tessellation associated with $U$ and $\lambda$).

The main result regarding this zeta function is the following theorem. %the proof of which we omit (cf.~Proposition \ref{prop:GeometricSequenceSufficient}.) 

\begin{theorem}
\label{thm:TessellationDimension}
Let $U \subset \R^m$ be a compact set with nonempty interior, which tessellates $\R^m$. Then, the upper box-counting dimension of a bounded infinite set $A$ in $\R^m$ is equal to the abscissa of convergence $D_{\calL(A,U,\lambda)}$ of the geometric zeta function \eqref{zeta_tess} of its tessellation fractal string. That is, 
\begin{align*}
\overline{\dim}_{B}A=D_{\calL(A,U,\lambda)}.
\end{align*}
\end{theorem}

\begin{proof}
Using Cauchy's criterion for convergence, we obtain that the series \eqref{zeta_tess} converges for all $s\in\C$ such that
$$
\limsup_{n\to\infty}m_n(\lambda)^{1/n}\lambda^{\operatorname{Re}(s)}<1,
$$
that is,
$$
\operatorname{Re}(s)>\frac{\log(\limsup_{n\to\infty}m_n(\lambda)^{1/n})}{\log\lambda^{-1}}.
$$
The series (\ref{zeta_tess}) diverges if we have the opposite inequality. Therefore, the abscissa of convergence of (\ref{zeta_tess}) is
\begin{equation}
\label{dimD}
D_{\calL(A,U,\lambda)}=\frac{\log(\limsup_{n\to\infty}m_n(\lambda)^{1/n})}{\log\lambda^{-1}}=\limsup_{n\to\infty}\frac{\log m_n(\lambda)}{\log\lambda^{-n}}.
\end{equation}
In light of \eqref{dimBA} and \eqref{dimD}, we deduce that $\overline\dim_BA=D_{\calL(A,U,\lambda)}$.
\end{proof}

\begin{remark}
The value of $D_{\calL(A,U,\lambda)}$ is independent of the choice of $\lambda\in(0,1)$, since by Theorem~\ref{thm:TessellationDimension} its
value is equal to $\overline\dim_BA$. In concrete applications we choose the basic shape $U$
and $\lambda\in(0,1)$ that are best suited to the geometry of $A$. For example, if $A$ is the triadic Cantor set, we take $U=[0,1]$ and $\lambda=1/3$, while for the
Sierpinski gasket we take $U$ to be an equilateral triangle and $\lambda=1/2$.
\end{remark}

\begin{example} 
Let $F$ be the 1-dimensional self-similar set from Examples \ref{eg:BoxCountingFractalStringOf1DFractal} and \ref{eg:BoxCountingZetaFunction1DFractal}. We define $U$ as the unit square $[0,1]^2$ and $\lambda=1/4$. Here, the scale $l_n(1/4)=1/4^n$ occurs with multiplicity $m_n(1/4)=9\cdot4^n$, defining the corresponding tessellation fractal string $\calL(F,U,1/4)$. For $\textnormal{Re}(s)>1$, the tessellation zeta function is given by 
$$
\zeta_{\calL(F,U,1/4)}(s) = \sum_{n=1}^\infty 9\cdot4^n\cdot 4^{-ns}=\frac9{4^{s-1}-1}.
$$
According to Theorem~\ref{thm:TessellationDimension}, the abscissa of convergence  $D_{\calL(F,U,1/4)}=1$ equal to the box-counting dimension of $F$. It follows that $\zeta_{\calL(F,U,1/4)}(s)$ has a meromorphic extension to all of $\C$ given by $9(4^{s-1}-1)^{-1}$. Furthermore, the set of tessellation complex dimensions is equal to the set $\calD_B$ of box-counting complex dimensions given in~ \eqref{eqn:1DBoxCountingComplexDimensions}. That is,
\[
\calD_{\calL(F,U,1/4)}=\calD_B=\left\{ 1 + i\frac{2\pi}{\log 4}z : z\in\Z \right\}.
\]
Note that the tessellation fractal string $\calL(F,U,1/4)$ is unbounded, in the sense that the series given by $\zeta_{\calL(F,U,1/4)}(1)$ is divergent.

Analogous results hold regarding the Cantor set $C$ and its (classical and box-counting) fractal strings, zeta functions, and complex dimensions. Further (higher-dimensional) examples will be studied in \cite{LapRoZu}.
\end{example}

\section{Distance and tube zeta functions}
\label{sec:DistanceAndTubeZetaFunctions}

In this section, we deal with a class of zeta functions introduced by the first author during the 2009 ISAAC Conference at the University of Catania in Sicily, Italy. More generally, the main results of this section are obtained in the forthcoming paper \cite{LapRaZu}, written by the first and third authors, along with Goran Radunovi\'c. We state here only some of the basic results, without attempting to work at the greatest level of generality. We refer to \cite{LapRaZu} for more general statements and additional results and illustrative examples.

The following definition can be found in \cite{LapRaZu}. 
\begin{definition} 
\label{def:DistanceZetaFunction}
Let $A \subset \R^m$ be bounded. The \emph{distance zeta function} of $A$, denoted $\zeta_d$, is defined by
\begin{align}
\label{eqn:DistanceZetaFunction}
\zeta_d(s)	&:=\int_{A_\varepsilon}d(x,A)^{s-m}dx
\end{align}
for $\textnormal{Re}(s)>D_d$, where $D_d=D_d(A)$ denotes the abscissa of convergence of the distance zeta function $\zeta_d$ and $\varepsilon$ is a fixed positive number. 
\end{definition}

\begin{remark}
\label{rmk:IndependentOfVarepsilon}
It is shown in \cite{LapRaZu} that 
%provided $|A|_m=0$ (for example, if $\overline{\dim}_BA<m$, a condition which is satisfied by most fractals of interest), then 
changing the value of $\varepsilon$ modifies the distance zeta function by adding an entire function to $\zeta_d$. Hence, the main properties of $\zeta_d$ do not depend on the choice of $\varepsilon>0$. As a result, this is also the case for $D_d$, the abscissa of convergence of $\zeta_d$ (cf. Theorem \ref{thm:DistanceZetaFunctionUpperBoxDimension}), and $\textnormal{res}(\zeta_d;D_d)$, the residue of $\zeta_d$ at $s=D_d$ (cf. Theorem \ref{thm:DistanceZetaFunctionAndMinkowskiContent}).
\end{remark}

The distance zeta function can be used as an effective tool in the computation of the box-counting dimensions of various subsets $A$ of some Euclidean space; see \cite{LapRaZu}. Indeed, one of the basic results concerning the distance zeta function is given in the following theorem, which is Theorem 1 in \cite{LapRaZu}. Note: unlike in Theorem \ref{thm:JohnnyDarko} above, we allow $A$ to be finite here.

\begin{theorem}
\label{thm:DistanceZetaFunctionUpperBoxDimension}
Let $A$ be a nonempty bounded subset of $\R^m$. Then $D_d=\overline{\dim}_{B}A$.
\end{theorem}

\begin{corollary}
\label{cor:holomorphic}
$\zeta_d$ is holomorphic in the half-plane $\textnormal{Re}(s)>\overline{\dim}_BA$. Furthermore, this open right half-plane is the largest one in which $\zeta_d(s)$ is holomorphic.  
\end{corollary}

\begin{remark}
\label{rmk:LowerBoxCountingDimensionFromDistanceZetaFunction}
We do not know if the value of the \emph{lower} box-counting dimension $\underline{\dim}_BA$ can be computed from the distance zeta function $\zeta_d$.
\end{remark}

It is shown in \cite{LapRaZu} that the distance zeta function represents a natural extension of the geometric zeta function $\zeta_{\calL}$ of a bounded (i.e., summable) fractal string $\calL=(\ell_j)_{j\in\N}$. Indeed, we can identify the string with an ordinary fractal string of the form $\Omega=\cup_{j=1}^\infty I_j$, where $I_j:=(a_{j+1},a_j)$ and $a_j:=\sum_{k\ge j}\ell_k$. Note that $|I_j|=\ell_j$. Defining $A=\{a_j\}_{j=1}^\infty\subset\R$, it is easy to see that $\zeta_d(s)=a(s)\zeta_{\calL}(s)+b(s)$, where $a(s)$ vanishes nowhere and $a(s)$ and $b(s)$ are explicit meromorphic functions in the complex plane with (typically) a pole at the origin. Hence, the zeta functions  $\zeta_\calL$ and $\zeta_d$ have the same abscissa of convergence. \emph{It follows that if $\calL$ is nontrivial} (\emph{i.e., has infinitely many lengths}), \emph{then}\footnote{If we allow $\calL$ to be trivial, then one should replace $D_d(A)$ and $D_\calL$ with $\max\{D_d(A),0\}$ and $\max\{D_\calL,0\}$, respectively, in Equation \eqref{eqn:DistanceDimensionEquivalence}.}
\begin{align}
\label{eqn:DistanceDimensionEquivalence}
D_d(A)=D_\calL=\overline{\dim}_BA=\overline{\dim}_B(\partial\Omega).
\end{align}

\subsection{Minkowski content and residue of the distance zeta function}
\label{sec:MinkowskiContentAndResidueOfTheDistanceZetaFunction}

A remarkable property of the distance zeta function is that its residue at $s=D_d$ is closely related to the $D_d$-dimensional Minkowski content of $A$; see \cite{LapRaZu}.

\begin{theorem}
\label{thm:DistanceZetaFunctionAndMinkowskiContent}
Let $A$ be a nonempty bounded set in $\R^m$.
Assuming that the distance zeta function can be meromorphically extended to a neighborhood of $s=D_d$ and $D_d<m$, then for its residue at $s=D_d$ we have that
\begin{align}
\label{inequality}
(m-D_d){\mink}_*^{D_d}\le \operatorname{res}(\zeta_d(s);D_d)\le(m-D_d){\mink}^{*D_d}.
\end{align}
If, in addition, $A$ is Minkowski measurable, it then follows that
\begin{align}
\label{eqn:MinkowskiContentAsResidue}
\operatorname{res}(\zeta_d(s);D_d)=(m-D_d){\mink}^{D_d}.
\end{align}
\end{theorem}

The last part of this result (namely, Equation \eqref{eqn:MinkowskiContentAsResidue}) generalizes the corresponding one obtained in \cite{LapvF6} in the context of ordinary fractal strings to the case of bounded sets in Euclidean spaces; see \cite{LapRaZu}.

\begin{example}
It can be shown that, in the case of the Cantor set $C$, we have strict inequalities in Equation \eqref{inequality}. Indeed, in this case $m=1$, $D_d=\log_{3}2$, and
$$
\operatorname{res}(\zeta_d(s);D_d)=\frac1{\log2}\,2^{-D_d},
$$
whereas the values of the lower and upper $D_d$-dimensional Minkowski contents have been computed in \cite[Theorem 2.16]{LapvF6} (as well as earlier in \cite{LapPo1}):
$$
{\mink}_*^{D_d}(A)=\frac1{D_d}\left(\frac{2D_d}{1-D_d} \right)^{1-D_d} ,\quad
{\mink}^{*D_d}(A)=2^{2-D_d}.
$$
This is a special case of an example in~\cite{LapRaZu} dealing with generalized Cantor sets. Generalized Cantor strings, which are a certain type of generalized fractal strings, and their (geometric and spectral) oscillations are studied in detail in \cite[Ch.~10]{LapvF6}.
\end{example}

\begin{remark}
An open problem is to determine whether there exists a set $A$ such that one of the inequalities in Equation \eqref{inequality} is strict and the other is an equality.
\end{remark}

\begin{remark}
\label{rmk:Resman}
According to a recent result due to Maja Resman in \cite{Res}, we know that if $A$ is Minkowski measurable, then the value of the \emph{normalized} $D_d$-\emph{dimensional Minkowski content} of a bounded set $A \subset \R^m$,\footnote{This choice of normalized Minkowski content is well known in the literature; see, e.g., \cite{Fdrr}.} defined by
\begin{align}
\label{eqn:DdMinkowskiContent}
\frac{{\mink}^{D_d}(A)}{\omega(m-D_d)},
\end{align}
is independent of the ambient dimension $m$. Here, for $t>0$, we let 
\[
\omega(t):=2\pi^{t/2}t^{-1}\Gamma(t/2)^{-1},
\]
where $\Gamma$ is the classic Gamma function. For any positive integer $k$, $\omega(k)$ is equal to the $k$-dimensional Lebesgue measure of the unit ball in $\mathbb{R}^k$. In other words, the value given in Equation \eqref{eqn:DdMinkowskiContent} is intrinsic to the set $A$ and hence independent of the embedding of $A$ in $\R^k$. Therefore, we may ask if the value of the normalized residue,
\begin{align*}
\frac{\mbox{res}(\zeta_d(s);D_d)}{(m-D_d)\,\omega(m-D_d)},
\end{align*}
is also independent of $m$. Combining the preceding two results (namely, Theorems \ref{thm:DistanceZetaFunctionUpperBoxDimension} and \ref{thm:DistanceZetaFunctionAndMinkowskiContent}), we immediately deduce that if $A$ is Minkowski measurable, then the answer is positive.
\end{remark}

\subsection{Tube zeta function}
\label{sec:TubeZetaFunction}

Given $\varepsilon>0$, it is also natural to introduce the following zeta function of a bounded set $A$ in $\R^m$, involving the tube around $A$ (which we view as the mapping $t\mapsto |A_t|$, for $0\leq t\leq\varepsilon$):\footnote{In the sequel, $|A_t|=|A_t|_m$ denotes the $m$-dimensional volume (Lebesgue measure) of $A_t$, the $t$-neighborhood of $A\subset\R^m$. In our earlier notation, we have $|A_t|=\vol_m(A_t)$.}
\begin{align}
\label{tube_zeta}
\tilde\zeta_A(s)	&=\int_0^\varepsilon t^{s-m-1}|A_t|\,dt,
\end{align}
for $\textnormal{Re}(s)$ sufficiently large, where $\varepsilon$ is a fixed positive number. Hence, $\tilde\zeta_A$ is called the \emph{tube zeta function} of $A$. Its abscissa of convergence is equal to $\overline{\dim}_{B}A$, which follows immediately from Theorems \ref{thm:DistanceZetaFunctionUpperBoxDimension} above and \ref{thm:DistanceAndTubeZetaFunctions} below. Tube zeta functions are closely related to distance zeta functions, as shown by the following result; see \cite{LapRaZu}.

\begin{theorem}
\label{thm:DistanceAndTubeZetaFunctions}
If $A\subset\R^m$ and $\operatorname{Re}(s)>\overline{\dim}_BA$, then for any $\varepsilon>0$,
\begin{equation}
\label{identity}
\zeta_d(s)	=\varepsilon^{s-m}|A_\varepsilon|+(m-s)\tilde\zeta_A(s).
\end{equation}
\end{theorem}

The proof of this result when $s\in\R$ follows, for example, from \cite[Theorem~2.9(a)]{Zu0}, and the proof of the cited theorem is based on integration by parts. By analytic continuation and in light of Corollary \ref{cor:holomorphic}, the identity \eqref{identity} is then extended to complex values of $s$ such that $\textnormal{Re}(s)>\overline{\dim}_BA$.

It follows from \eqref{identity} that the abscissae of convergence of the zeta functions $\zeta_d$ and $\tilde\zeta_A$ are the same, and therefore also coincide with $\overline{\dim}_BA$. This identity, \eqref{identity}, extends to the  $m$-dimensional case \cite[identity (13.129) in Lemma 13.110, p.~442]{LapRaZu}, which has been formulated in the context of $p$-adic fractal strings and ordinary (real) fractal strings. (See also \cite{LapLuvF1}.) Using \eqref{identity} and Theorem \ref{thm:DistanceZetaFunctionAndMinkowskiContent}, it is easy to derive the following consequence; see \cite{LapRaZu}.

\begin{corollary} 
\label{cor_res}
If $D=\dim_BA$ exists, $D<m$, and there exists a meromorphic extension of $\tilde\zeta_A(s)$ to a neighborhood of $s=D$, then
\begin{align*}
&{\mink}_*^D \le\operatorname{res}(\tilde\zeta_A(s);D)\le {\mink}^{*D}.
\end{align*}
In particular, if $A$ is Minkowski measurable, then 
\begin{align*}
\operatorname{res}(\tilde\zeta_A(s);D)	&={\mink}^D.
\end{align*}
\end{corollary}

As we can see, the tube zeta function is ideally suited to study the Minkowski content.

In Corollary~\ref{cor_res}, we have assumed the existence of a meromorphic extension of the tube zeta function to a neighborhood of $s=D$. This condition can be ensured under fairly general conditions. We provide a result from~\cite{LapRaZu} dealing with the case of Minkowski measurable sets. Non-Minkowski measurable sets can be treated as well; see \cite{LapRaZu}, along with Remark \ref{rmk:NonMinkowskiAnalogs} and Theorem \ref{thm:NonMinkowskiMeasurableCase}. 

\begin{theorem}[Minkowski measurable case]
\label{mer}
Let $A$ be a subset of $\R^m$ such that there exist  $\alpha>0$, $M\in(0,\infty)$ and $D\in[0,m]$,  satisfying
\begin{align}
\label{eqn:MinkowskiMeasurableCase}
|A_t| &= t^{m-D}\left(M+O(t^\alpha)\right)\quad\mbox{as $t\to 0^+$.}
\end{align}
Then $A$ is Minkowski measurable, $\operatorname{dim}_BA=D$, and $\mink^D(A)=M$. Furthermore, the tube zeta function $\tilde{\zeta}_A(s)$ has for abscissa of convergence  $D(\tilde\zeta_A)=D$, and it admits a \emph{(}necessarily unique\emph{)} meromorphic continuation {\rm({\it at least})} to the half-plane $\{\operatorname{Re}(s)>D-\alpha\}$. The only pole in this half-plane is $s=D$; it is simple, and $\operatorname{res}(\tilde\zeta_A;D)=M$.
\end{theorem}

An analogous result holds also for the distance function of $A$. Theorem~\ref{mer} shows that  the relevant information concerning the possible existence of a nontrivial meromorphic extension of the tube (or the distance) zeta function associated with~$A$, is encoded in the second term of the asymptotic expansion (as $t \to 0^+$) of the tube function $t\mapsto|A_t|$. Various extensions of this result and examples can be found in~\cite{LapRaZu}, as we next discuss. For example, in Theorem \ref{mer}, the conclusion for $\zeta_d$ would be the same except for the fact that $\operatorname{res}(\zeta_d(s);D)=(m-D)M$ (and $D=D_d(A)$, in our earlier notation from Equation \eqref{eqn:DistanceDimensionEquivalence}, for instance). 

\begin{remark}
\label{rmk:NonMinkowskiAnalogs}
In \cite{LapRaZu}, one can find suitable analogs of Theorem \ref{mer} for non-Minkowski measurable sets, both in the case where the underlying scaling behavior of $|A_t|$ is log-periodic (as for the Cantor set or the Sierpinski gasket and carpet, for example) and in more general, non-periodic situations. Furthermore, in that case, the (visible) complex dimensions of $A$ (i.e., the poles of $\tilde{\zeta}_A$ in the half-plane $\{\textnormal{Re}(s)>D-\alpha\}$, with $\alpha>0$, are also determined, and shown to consist of a vertical infinite arithmetic progression located on the `critical line' $\{\textnormal{Re}(s)=D\}$ (much as for the Cantor string and other lattice self-similar strings); see Theorem \ref{thm:NonMinkowskiMeasurableCase} below for a typical sample theorem. In \cite{LapRaZu}, the case of so-called `quasi-periodic' fractals is also considered in this and related contexts. In the latter situation, one appeals, in particular, to well-known (and rather sophisticated) number theoretic theorems asserting the transcendentality (and hence, the irrationality) of certain expressions. It is noteworthy that (in light of Theorem \ref{thm:DistanceAndTubeZetaFunctions}) all of these results concerning $\tilde{\zeta}_A$ have precise counterparts for the distance zeta function $\zeta_d$; see \cite{LapRaZu}.  
\end{remark}

In order to state more precisely a sample theorem in the non-Minkowski measurable case, we introduce the following hypothesis (LP) (log-periodic) and notation:

\begin{itemize} 
\item[(LP)] Let $A\subset\R^m$ be a bounded set such that there exists $D\geq 0, \alpha>0$ and a periodic function $G:\R\to [0,\infty)$ with minimal period $T>0$, satisfying
\end{itemize}
\begin{align}
\label{eqn:NonConstantPeriodic}
|A_t| &= t^{m-D}(G(\log{t^{-1}})+O(t^\alpha)), \quad \textnormal{as } t\to 0^+.
\end{align}
In the sequel, let
\begin{align}
\label{eqn:NonConstantPeriodicIntegral}
\hat{G}_0(t) &= \int_0^T e^{-2\pi it\tau}G(\tau)\,d\tau, \quad \textnormal{for } t\in\R.
\end{align}
Note that $G$ is nonconstant since we have assumed it to have a positive minimal period, and that $\hat{G}_0$ is (essentially) the Fourier transform of the cut-off function $G_0$ of $G$ to $[0,T]$.

\begin{theorem}[Non-Minkowski measurable case]
\label{thm:NonMinkowskiMeasurableCase}
Assume that the bounded set $A\subset\R^m$ satisfies the log-periodicity property \emph{(LP)} above.  Then $\dim_BA$ exists, $\dim_BA=D$, $G$ is continuous and\footnote{It follows that the restriction of $G$ to $[0,T]$ has for range $[\calM^D_*(A),\calM^{*D}(A)]$.}
\begin{align}
\label{eqn:UpperLowerMinkowskiContentNonMinkCase}
\calM^D_*(A)&=\min{G}, \quad \calM^{*D}(A)=\max{G}.
\end{align}
Hence, $A$ is not Minkowski measurable and \emph{(}provided $G>0$\emph{)} is non-degenerate, i.e., $0<\calM^D_*(A)<\calM^{*D}(A)<\infty$.

Furthermore, $\tilde{\zeta}_A$ admits a \emph{(}necessarily unique\emph{)} meromorphic continuation to the half-plane $\{\textnormal{Re}(s)>D-\alpha\}$ with only poles 
\begin{align}
\label{eqn:PolesNonMinkowskiCase}
\mathcal{P}(\tilde{\zeta}_A)	&=\left\lbrace s_k:=D+\frac{2\pi}{T}ik: \hat{G}_0\left(\frac{k}{T}\right)\neq 0, k\in\Z\right\rbrace,
\end{align}
where $\hat{G}_0$ is given by \eqref{eqn:NonConstantPeriodicIntegral}. These poles are all simple and\footnote{In addition, $|\textnormal{res}(\tilde{\zeta}_A;s_k)|\leq\frac{1}{T}\int_0^TG(\tau)\,d\tau$ for all $k\in\Z$, and $\textnormal{res}(\tilde{\zeta}_A;s_k)\to 0$ as $|k|\to\infty$.}
\begin{align}
\label{eqn:ResiduesNonMinkCase}
\textnormal{res}(\tilde{\zeta}_A;s_k)	&= \frac{1}{T}\hat{G}_0\left(\frac{k}{T}\right), \quad \textnormal{for all }k\in\Z.
\end{align}
In particular, for $k=0$, we have that 
\begin{align}
\label{eqn:ResidueAsIntegralNonMinkCase}
\textnormal{res}(\tilde{\zeta}_A;D)	&=\frac{1}{T}\int_0^TG(\tau)\,d\tau.
\end{align}

Finally, $A$ admits an average Minkowski content $\calM^D_{av}$ \emph{(}which lies in $(0,\infty)$ if $G>0$\emph{)} that is also given by \eqref{eqn:ResidueAsIntegralNonMinkCase}.\footnote{This average Minkowski content is defined exactly as in \cite[Definition~8.29,\S 8.4.3]{LapvF6}, except for the fact that $1-D$ is replaced with $m-D$.}   
\end{theorem} 

\begin{remark}
\label{rmk:CounterpartTheoremForDistanceZetaFunctions}
As was alluded to in Remark \ref{rmk:NonMinkowskiAnalogs}, Theorem \ref{thm:NonMinkowskiMeasurableCase} has a precise counterpart for the distance zeta function $\zeta_d$. In fact, under the same assumption (LP), the same conclusions as in Theorem \ref{thm:NonMinkowskiMeasurableCase} hold, except for the fact that the counterpart of \eqref{eqn:ResiduesNonMinkCase} and \eqref{eqn:ResidueAsIntegralNonMinkCase}, respectively, reads as follows (with $D=D_{d}(A)$, as in Equation \eqref{eqn:DistanceDimensionEquivalence} and Theorems \ref{thm:DistanceZetaFunctionAndMinkowskiContent} and \ref{thm:NonMinkowskiMeasurableCase}):
\begin{align}
\label{eqn:ResiduesNonMinkCaseDistance}
\textnormal{res}(\zeta_d;s_k)	&= \frac{m-s_k}{T}\hat{G}_0\left(\frac{k}{T}\right), \quad \textnormal{for all }k\in\Z,
\end{align}
and 
\begin{align}
\label{eqn:ResidueAsIntegralNonMinkCaseDistance}
\textnormal{res}(\zeta_d;D)	&=(m-D)\frac{1}{T}\int_0^TG(\tau)\,d\tau,
\end{align}
so that $\calM^D_{av}(A)$, the average Minkowski content of $A$, is now given by 
\begin{align}
\label{eqn:AverageMinkowskiContent}
\calM^D_{av}(A)	&=\frac{1}{T}\int_0^TG(\tau)\,d\tau=\frac{1}{m-D}\textnormal{res}(\zeta_d;D).
\end{align}
\end{remark}

\begin{example}
\label{eg:IllustratingNonMinkCase}
Theorem \ref{thm:NonMinkowskiMeasurableCase} (and its counterpart for $\zeta_d$ discussed in Remark \ref{rmk:CounterpartTheoremForDistanceZetaFunctions}) can be illustrated, for instance, by the Cantor set or string (see \cite{LapRaZu} and compare with \cite[\S 1.1.2]{LapvF6}, including Figures 1.4 and 1.5) and the Sierpinski carpet (as well as, similarly, by the Sierpinski gasket). For the Sierpinski carpet $A$, as discussed in \cite{LapRaZu}, we have $D=\log_{3}8$, $\alpha=D-1$, and $T=\log{3}$. Hence, both $\tilde{\zeta}_A$ and $\zeta_d$ have a meromorphic continuation to $\left\lbrace \textnormal{Re}(s)>1\right\rbrace$, with set of (simple) poles (the visible complex dimensions of $A$) given by
\begin{align}
\label{eqn:IllustrationNonMinkCaseComplexDimensions}
\mathcal{P}(\tilde{\zeta}_A)	&= \mathcal{P}(\zeta_d) = \left\lbrace D+\frac{2\pi}{T}ik: \hat{G}_0\left(\frac{k}{T}\right)\neq 0, k\in\Z\right\rbrace.
\end{align}
Actually, in this case, both $\tilde{\zeta}_A$ and $\zeta_d$ can be meromorphically extended to all of $\C$.
\end{example}

A number of other results concerning the existence of meromorphic continuation of $\tilde{\zeta}_A$ and $\zeta_d$ (as well as $\zeta_\calL$, in the case of fractal strings) and the resulting structure of the poles (the visible complex dimensions) can be found in \cite{LapRaZu}, under various assumptions on the bounded set $A\subset\R^m$ (or on the fractal string $\calL$). Moreover, a number of other applications of the distance and tube zeta functions are provided in \cite{LapRaZu}, in order to study various classes of fractals, including fractal chirps and `zigzagging' fractals.

\begin{remark}
The box-counting zeta function $\zeta_B$ of a set $A \subset \R^m$ given by Definition \ref{def:BoxCountingZetaFunction} is closely related to the tube zeta function $\tilde\zeta_A$. To see this, it suffices to perform the change of variables $x=t^{-1}$ in Equation \eqref{tube_zeta} and compare with Corollary~\ref{cor:EquivalentZetaFunctions}. Note that for $x>0$, we have  (under suitable hypotheses) $|A_{1/x}|\asymp x^{-m}N_B(A,x)$ as $x\to\infty$. Here, $N_B(A,x)$ is defined as the number of $x^{-1}$-mesh cubes that intersect $A$; see (iv) in Remark \ref{rmk:VariousCountingFunctions}. It is clear, however, that these two zeta functions are in general not equal to each other. Moreover, we do not know if the corresponding two sets of complex dimensions of $A$, associated with these two zeta functions, coincide.
\end{remark}

Various generalizations of the notion of distance zeta function are possible. One of them, which is especially interesting, deals with zeta functions associated to relative fractal drums. By a \emph{relative fractal drum}, introduced in \cite{LapRaZu}, we mean an ordered pair $(A,\Omega)$, where $A$ is an arbitrary nonempty subset of $\R^m$, and $\Omega$ is an open subset such that $A_\varepsilon$ contains $\Omega$ for some positive $\varepsilon$, and the $m$-dimensional Lebesgue measure of $\Omega$ is finite. (Note that both $A$ and $\Omega$ are now allowed to be unbounded.) The corresponding \emph{relative zeta function} (or the distance zeta function of the relative fractal drum), also introduced in \cite{LapRaZu}, is defined in much the same way as in Equation \eqref{eqn:DistanceZetaFunction}:
\begin{align*}
%\label{eqn:DistanceZetaFunctionOfThePair}
\zeta_d(s;A,\Omega)	&:=\int_{\Omega}d(x,A)^{s-m}dx.
\end{align*}

It is possible to show that the abscissa of convergence of the relative zeta function is equal to the relative box dimension $\overline\dim_B(A,\Omega)$; see \cite{LapRaZu} for details and illustrative examples.\footnote{We caution the reader that, in general, the relative upper box dimension (defined as an infimum over $\alpha \in \R$ instead of over $\alpha\geq 0$) may be negative and that there are even cases where it is equal to $-\infty$ (e.g., when $\textnormal{dist}(A,\Omega)>0)$; see \cite{LapRaZu}.}

\begin{remark}
\label{rmk:StandardFractalDrum}
For the standard notion of fractal drum, which corresponds to the choice $A=\partial\Omega$ with $\Omega$ bounded, we refer, e.g., to \cite{Lap1,Lap2} along with \cite[\S 12.5]{LapvF6} and the relevant references therein. 
\end{remark}

\begin{remark}
\label{rmk:RelativeFractalDrums}
It is easy to see that the notion of relative fractal drum $(A,\Omega)$ is a natural extension of the notion of fractal string ${\calL}=\{\ell_j\}$. Indeed, for a given (standard) fractal string ${\calL}=\{\ell_j\}$, it suffices to define $A=\{a_j\}$, where $a_j:=\sum_{k\ge j}\ell_k$ and $\Omega=\cup_{k\ge1} (a_{k+1},a_k)$. We caution the reader that the notion of generalized fractal string already exists but does not coincide with the notion of relative fractal drum. Specifically, in \cite[Ch.~4]{LapvF6}, a generalized fractal string is defined to be a locally positive or locally complex measure on $(0,\infty)$ supported on a subset of $(x_0,\infty)$, for some positive real number $x_0$.
\end{remark}

\section{Summary of results and open problems}
\label{sec:SummaryOfResultsAndOpenProblems}

For a bounded infinite set $A$, recall that $\overline{\dim}_{B}A$ denotes the upper box-counting dimension of $A$ given by Equation \eqref{eqn:UpperBoxDimension}, $\overline{\dim}_{M}A$ denotes the upper Minkowski dimension of $A$ given by Equation \eqref{eqn:UpperMinkowskiDimension}, $D_B$ denotes the abscissa of convergence of the box-counting zeta function $\zeta_B$ of $A$ given in Definition \ref{def:BoxCountingZetaFunction}, $\rho_\calL$ denotes the order of the geometric counting function $N_\calL$ given by Equation \eqref{eqn:OrderOfCountingFunction} where $\calL=\calL_B$, $D_N$ denotes the value corresponding to the (asymptotic) growth rate of $N_\calL$ given by Equation \eqref{eqn:CountingFunctionD}, and $D_d$ is the abscissa of convergence of the distance zeta function $\zeta_d$ given in Definition \ref{def:DistanceZetaFunction}. Furthermore, let $D_t$ denote the abscissa of convergence of the tube zeta defined in Equation \eqref{tube_zeta}.

The following theorem summarizes our main result (as stated in Theorem \ref{thm:Summary} of the introduction), which pertains to the determination of the box-counting dimension of a bounded infinite set. (Recall that the equalities $\overline{\dim}_{B}A=D_d=D_t$ are established in \cite{LapRaZu}; see Theorem \ref{thm:DistanceZetaFunctionUpperBoxDimension} and the comment following Theorem \ref{thm:DistanceAndTubeZetaFunctions} above.)

\begin{theorem}
\label{thm:MainResult}
Let $A$ be a bounded infinite subset of $\R^m$ and let $\calL=\calL_B$ be the corresponding box-counting fractal string. Then the following equalities hold\emph{:}
\begin{align*}
\overline{\dim}_{B}A &= \overline{\dim}_{M}A = D_B = \rho_\calL = D_N = D_d=D_t.
\end{align*} 
\end{theorem}

\begin{proof}
The classic equality $\overline{\dim}_{B}A = \overline{\dim}_{M}A$ is established in \cite{Falc}. The equality $\overline{\dim}_{M}A=D_B=D_N$ follows from Theorem \ref{thm:JohnnyDarko}. The equality $\rho_\calL=D_N$ follows from Remark \ref{rmk:DimensionEqualsOrder}. Finally, as was recalled just above, the equalities $\overline{\dim}_{M}A=D_d$ and $\overline{\dim}_{M}A=D_t$ are established in \cite{LapRaZu}; see Theorem \ref{thm:DistanceZetaFunctionUpperBoxDimension} and the comment following Theorem \ref{thm:DistanceAndTubeZetaFunctions}. These last two equalities are valid whether or not $A$ is infinite.
\end{proof}

Recall that, as stated in Definition \ref{def:BoxCountingFunction}, $N_B(A,x)$ denotes the maximum number of disjoint balls of radius $x^{-1}$ centered in $A$. In this setting and for $\varepsilon>0$ we have
\begin{align}
\label{eqn:LowerBoundForVolume}
B_m \varepsilon^m N_B(A,\varepsilon^{-1}) \leq |A_\varepsilon|,
\end{align}
where $A_\varepsilon$ is the $\varepsilon$-neighborhood of $A$, $|A_\varepsilon|=|A_\varepsilon|_m$ $B_m$ is the $m$-dimensional volume of a ball in $\R^m$ with unit radius, and $0<\varepsilon<x_1^{-1}$, where $x_1^{-1}$ is given by Proposition \ref{prop:BoxCountingJumps}. 

Motivated by Equation \eqref{eqn:LowerBoundForVolume} and Theorem \ref{thm:CountingFunctionAndComplexDimensions}, we propose the following open problem (which is stated rather roughly here).

\begin{openproblem}
\label{op:BoxCountingAndCountingFormulas}
Let $A$ be a bounded infinite subset of $\R^m$ with box-counting fractal string $\calL_B$. Assume suitable growth conditions on $\zeta_B$ \emph{(}such as the languidity of $\zeta_B$ on an appropriate window, see \emph{\cite[Chs.~5 \& 8]{LapvF6})} and assume for simplicity that all of the complex dimensions are simple \emph{(}i.e., are simple poles of $\zeta_B$\emph{)}. Then, as $\varepsilon\to 0^+$, compare the quantities
\begin{align}
\label{eqn:VolumeFormulaOpenProblem}
|A_\varepsilon|, \quad \varepsilon^m N_{\calL_B}(\varepsilon^{-1}), \quad \textnormal{and } \, \varepsilon^m \left(\sum_{\omega \in \calD_B}
\frac{\varepsilon^{-\omega}}{\omega}\textnormal{res}(\zeta_B(s);\omega) +  R(\varepsilon^{-1})\right), 
\end{align}
where $R(\varepsilon^{-1})$ is an error term of small order. 
\end{openproblem}

If one were to provide a more precise version of the above open problem and solve it, one might consider pursuing a generalization of Theorem \ref{thm:CriterionForMinkowskiMeasurability} in the spirit of the theory of complex dimensions of fractal strings, as described in \cite{LapvF6}, and of its higher-dimensional counterpart in \cite{LapPe3,LapPe2,LapPeWin}. Naturally, the clarified version of this open problem would consist of replacing the implicit `approximate equalities' in Equation \eqref{eqn:VolumeFormulaOpenProblem} with true equalities, modulo suitable modifications  and under appropriate hypotheses.

Analogously (but possibly more accurately), in light of the results from \cite{LapRaZu} discussed in Section \ref{sec:DistanceAndTubeZetaFunctions}, as well as from the results about fractal tube formulas obtained in \cite[Ch.~~8]{LapvF6} for fractal strings and in \cite{LapPe3,LapPe2} and especially \cite{LapPeWin} in the higher-dimensional case (for fractal sprays and self-similar tilings),\footnote{A survey of the results of \cite{LapPe3,LapPe2,LapPeWin} can be found in \cite[\S 13.1]{LapvF6}.} we propose the following open problem. (A similar problem can be posed for the tube zeta function $\tilde\zeta_A$ discussed in Section \ref{sec:TubeZetaFunction}.)

\begin{openproblem}
\label{op:VolumeTubularNeighborhood}
Let $A$ be a bounded subset of $\R^m$ with distance zeta function $\zeta_d$. Under suitable growth assumptions on $\zeta_d$ \emph{(}such as the languidity of $\zeta_B$ on an appropriate window, see \emph{\cite[Chs.~5 \& 8]{LapvF6})}, and assuming for simplicity that all of the corresponding complex dimensions are simple, calculate the volume of the tubular neighborhood of $A$ in terms of the complex dimensions of $A$ \emph{(}defined here as the poles of the meromorphic continuation of $\zeta_d$ union the `integer dimensions' \{0,1,\ldots,m\}\emph{)} and the associated residues. 

Moreover, even without assuming that the complex dimensions are simple, express the resulting fractal tube formula as a sum of residues of an appropriately defined `tubular zeta function' \emph{(}in the sense of \emph{\cite{LapPe3,LapPe2,LapPeWin,LapPeWi2})}.
\end{openproblem}

%%%%%%%%%
\section*{Acknowledgments}
%%%%%%%%%

The first author (M.~L.~Lapidus) would like to thank the Institut des Hautes Etudes Scientifiques (IHES) in Bures-sur-Yvette, France, for its hospitality while he was a visiting professor in the spring of 2012 and this paper was completed.

The authors would like to thank our anonymous referees. They provided helpful comments and suggestions in their very thorough reviews of a preliminary version of this paper.

In closing, the authors would like to thank the Department of Mathematics at the University of Messina and the organizers of the Permanent International Session of Research Seminars (PISRS), especially Dr.\ David Carfi, for their organization of and invitations to speak in the First International Meeting of PISRS, Conference 2011: Analysis, Fractal Geometry, Dynamical Systems, and Economics. The authors' participation in this conference led directly to collaboration on the development of the new results presented in the paper, namely the results pertaining to box-counting fractal strings and box-counting zeta functions.

%%%%%%%%%%%%%%%%%%
\bibliographystyle{amsalpha}
%%%%%%%%%%%%%%%%%%
%\addcontentsline{toc}{section}{References}
%\section{References}

%\include{TechniquesBCD_bib}

\end{document}